\documentclass[aps,prb,twocolumn, superscriptaddress, showpacs]{revtex4-1}
\usepackage{amsmath}
\usepackage{amssymb,amsfonts}
\usepackage{graphicx}
\usepackage{mathtools}
\usepackage{relsize}
\usepackage{lmodern}
\usepackage{slantsc}
\usepackage{scalefnt}
\usepackage{subfigure}
\usepackage{dsfont}
\usepackage{xspace} 
\usepackage{boxedminipage}
\usepackage{ifpdf}
\usepackage{pgffor}
\usepackage{xifthen}
\usepackage{tensor}
\usepackage{array}
\usepackage{tikz}
\usetikzlibrary{shapes.geometric}
\usepackage{color}
\usetikzlibrary{arrows,matrix,calc,scopes,decorations.markings}
\allowdisplaybreaks[1]

\newtheorem{theorem}{Theorem}
\newtheorem{proposition}[theorem]{Proposition}

\newenvironment{proof}[1][Proof]{\noindent\textbf{#1.} }{\ \rule{0.5em}{0.5em}}
\newcommand{\be}{\begin{equation}}
\newcommand{\ee}{\end{equation}}
\newcommand{\bse}{\begin{subequations}}
\newcommand{\ese}{\end{subequations}}

\newcommand{\Z}{\mathbb{Z}}
\newcommand{\ii}{\mathrm{i}}
\newcommand{\e}{\mathrm{e}}

\newcommand{\T}{\mathcal{T}}
\newcommand{\A}{\mathcal{A}}
\newcommand{\U}{\mathcal{U}}
\newcommand{\bpm}{\begin{pmatrix}}
\newcommand{\epm}{\end{pmatrix}}
\newcommand{\bmm}{\begin{matrix}}
\newcommand{\emm}{\end{matrix}}

\newcommand{\Blangle}{\Biggl\langle\bmm} 
\newcommand{\BLvert}{\Biggl\vert\bmm} 
\newcommand{\Bvert}{\emm\Biggr\vert\bmm} 
\newcommand{\Brangle}{\emm\Biggr\rangle}

\newcommand{\x}{\times}
\newcommand{\II}{$\text{Ising}\times\overline{\text{Ising}}$\xspace}
\newcommand{\is}{1\sigma}
\newcommand{\si}{\sigma 1}
\newcommand{\ip}{1\psi}
\newcommand{\psii}{\psi 1}
\newcommand{\pp}{\psi\psi}
\newcommand{\psis}{\psi\sigma}
\newcommand{\spsi}{\sigma\psi}
\renewcommand{\ss}{\sigma\sigma}
\newcommand{\tchi}{\tilde\chi}
\newcommand{\diag}{\mathrm{Diag}}
\newcommand{\vac}{\gamma}

\newenvironment{smallarray}[1]
 {\null\,\vcenter\bgroup\scriptsize
  \arraycolsep=.13885em
  \hbox\bgroup$\array{@{}#1@{}}}
 {\endarray$\egroup\egroup\,\null}

\newcommand{\vlc}[6]{
\bigl[\hspace{-0.2em}
\begin{smallarray}{cc|c}
    #1 & #2 & #3\\ [0.2em]
    #4 & #5 & #6
\end{smallarray}\hspace{-0.2em}\bigr] }

\newcommand{\sixj}[6]{
\bigl\{\hspace{-0.2em}
\begin{smallarray}{cc|c}
    #1 & #2 & #3\\ [0.2em]
    #4 & #5 & #6
\end{smallarray}\hspace{-0.2em}\bigr\} }

\newcommand{\Fm}[7][0]{
\ifthenelse{\NOT \equal{#1}{0}}{
\left[F^{#2#3#4}_{#5}\right]_{#6#7}
}
{\bigl[F^{#2#3}_{#4#5}\bigr]^{#6}_{#7}}
}

\makeatletter
\newcommand*{\Relbarfill@}{\arrowfill@\Relbar\Relbar\Relbar}
\newcommand*{\xeq}[2][]{\ext@arrow 0055\Relbarfill@{#1}{#2}}
\makeatother

\tikzset{->-/.style={decoration={
  markings,
  mark=at position .5 with {\arrow{>}}},postaction={decorate}}}
\tikzset{-<-/.style={decoration={
  markings,
  mark=at position .5 with {\arrow{<}}},postaction={decorate}}}
\tikzset{arrow data/.style 2 args={%
      decoration={%
         markings,
         mark=at position #1 with \arrow{#2}},
         postaction=decorate}
      }%

\newcommand{\Yup}[5][0]{
\begin{tikzpicture}[scale=0.5,baseline={(current bounding box.center)}]
  \ifthenelse{\NOT \equal{#2}{-1} \AND \NOT \equal{#2}{-2}}  {
    \draw[line cap=round]  (150:1) -- (0,0) ; 
    \node at(-0.5,-0.05) {\scalebox{0.7}{$#2$}};
    \ifthenelse{\equal{#1}{1}} {
      \draw [->,>=stealth',line width=0.01pt] (0,0) -- (150:0.5);
      }{}
    }{
    \ifthenelse{\equal{#2}{-1}} {
    \draw[dotted,line cap=round] (150:1) -- (0,0) ;} 
    {}
    \ifthenelse{ \equal{#2}{-2}} {
    \draw[dashed,line cap=round] (150:1) -- (0,0);
    }{}
    }

  \ifthenelse{\NOT \equal{#3}{-1} \AND \NOT \equal{#3}{-2}}  {
    \draw[line cap=round]  (30:1)  -- (0,0) ; 
    \node at(0.5,-0.05) {\scalebox{0.7}{$#3$}};
    \ifthenelse{\equal{#1}{1}} {
       \draw [->,>=stealth',line width=0.01pt] (0,0) -- (30:0.5);
          }{}
    }{
    \ifthenelse{ \equal{#3}{-1}} {
    \draw[dotted,line cap=round] (30:1) -- (0,0) ;} 
    {}
    \ifthenelse{ \equal{#3}{-2}} {
    \draw[dashed,line cap=round] (30:1) -- (0,0);
    }{}
    }    

  \ifthenelse{\NOT \equal{#4}{-1} \AND \NOT \equal{#4}{-2}}  {
    \draw[line cap=round] (0,-1) -- (0,0); 
    \node at(0.35,-0.6) {\scalebox{0.7}{$#4$}};  
    \ifthenelse{\equal{#1}{1}} {
      \draw [->,>=stealth',line width=0.01pt] (0,-1) -- (0,-0.4);
          }{}    
    }{
    \ifthenelse{\equal{#4}{-1}} {
    \draw[dotted,line cap=round] (0,-1) -- (0,0) ;} 
    {}
    \ifthenelse{\equal{#4}{-2}} {
    \draw[dashed,line cap=round] (0,-1) -- (0,0);
    }{}
    } 
  \node[above] at(0,-0.1) {\scalebox{0.7}{$#5$}};
  \end{tikzpicture}
}

\newcommand{\Ydown}[5][0]{
\begin{tikzpicture}[scale=0.5,baseline={(current bounding box.center)}]
  \ifthenelse{\NOT \equal{#2}{-1} \AND \NOT \equal{#2}{-2}}  {
    \draw[line cap=round]  (210:1) -- (0,0) ; 
    \node at(-0.5,0.01) {\scalebox{0.7}{$#2$}};
    \ifthenelse{\equal{#1}{1}} {
      \draw [->,>=stealth',line width=0.01pt] (210:0.5) -- (210:0.4);
      }{}
    }{
    \ifthenelse{\equal{#2}{-1}} {
    \draw[dotted,line cap=round] (210:1) -- (0,0) ;} 
    {}
    \ifthenelse{ \equal{#2}{-2}} {
    \draw[dashed,line cap=round] (210:1) -- (0,0);
    }{}
    }

  \ifthenelse{\NOT \equal{#3}{-1} \AND \NOT \equal{#3}{-2}}  {
    \draw[line cap=round]  (330:1)  -- (0,0) ; 
    \node at(0.5,0.05) {\scalebox{0.7}{$#3$}};
    \ifthenelse{\equal{#1}{1}} {
       \draw [->,>=stealth',line width=0.01pt] (330:0.5) -- (330:0.4);
          }{}
    }{
    \ifthenelse{ \equal{#3}{-1}} {
    \draw[dotted,line cap=round] (330:1) -- (0,0) ;} 
    {}
    \ifthenelse{ \equal{#3}{-2}} {
    \draw[dashed,line cap=round] (330:1) -- (0,0);
    }{}
    }    

  \ifthenelse{\NOT \equal{#4}{-1} \AND \NOT \equal{#4}{-2}}  {
    \draw[line cap=round] (0,0) -- (0,1); 
    \node at(0.4,0.6) {\scalebox{0.7}{$#4$}};  
    \ifthenelse{\equal{#1}{1}} {
      \draw [->,>=stealth',line width=0.01pt] (0,0) -- (0,0.5);
          }{}    
    }{
    \ifthenelse{\equal{#4}{-1}} {
    \draw[dotted,line cap=round] (0,0) -- (0,1) ;} 
    {}
    \ifthenelse{\equal{#4}{-2}} {
    \draw[dashed,line cap=round] (0,0) -- (0,1);
    }{}
    } 
  \node[below] at(0, 0.1) {\scalebox{0.7}{$#5$}};  
  \end{tikzpicture}
}

\newcommand{\Hgraph}[6][0]{
  \begin{tikzpicture}[scale=0.8,baseline=-3]
    \coordinate (ul) at (-0.375,0.75);
    \coordinate (lr) at (0.375,-0.75);
    \coordinate (ll) at (-0.375,-0.75);
    \coordinate (ur) at (0.375,0.75);    
    \coordinate (l) at ($ (ll) ! 0.6 ! (ul) $);
    \coordinate (r) at ($ (lr) ! 0.4 ! (ur) $);

  \ifthenelse{\NOT \equal{#2}{-1} \AND \NOT \equal{#2}{-2}}  {
    \draw[line cap=round] (l) -- (ul); 
    \node[left] at($ (ll) ! 0.85 ! (ul) $) {\scalebox{0.8}{$#2$}};  
    \ifthenelse{\equal{#1}{1}} {
      \draw[->,>=stealth',line width=0.01pt] (l) -- ($ (ll) ! 0.85 ! (ul) $);
          }{}    
    }{
    \ifthenelse{\equal{#2}{-1}} {
    \draw[dotted,line cap=round] (l) -- (ul) ;} 
    {}
    \ifthenelse{\equal{#2}{-2}} {
    \draw[dashed,line cap=round] (l) -- (ul);
    }{}
    }    
    
  \ifthenelse{\NOT \equal{#3}{-1} \AND \NOT \equal{#3}{-2}}  {
    \draw[line cap=round] (r) -- (ur); 
    \node[right] at($ (lr) ! 0.85 ! (ur) $) {\scalebox{0.8}{$#3$}};  
    \ifthenelse{\equal{#1}{1}} {
      \draw[->,>=stealth',line width=0.01pt] (r) -- ($ (lr) ! 0.85 ! (ur) $);
          }{}    
    }{
    \ifthenelse{\equal{#3}{-1}} {
    \draw[dotted,line cap=round] (r) -- (ur) ;} 
    {}
    \ifthenelse{\equal{#3}{-2}} {
    \draw[dashed,line cap=round] (r) -- (ur);
    }{}
    }

  \ifthenelse{\NOT \equal{#4}{-1} \AND \NOT \equal{#4}{-2}}  {
    \draw[line cap=round] (ll) -- (l); 
    \node[left] at($ (ll) ! 0.2 ! (ul) $) {\scalebox{0.8}{$#4$}};  
    \ifthenelse{\equal{#1}{1}} {
      \draw[->,>=stealth',line width=0.01pt] (ll) -- ($ (ll) ! 0.3 ! (ul) $);
          }{}    
    }{
    \ifthenelse{\equal{#4}{-1}} {
    \draw[dotted,line cap=round] (ll) -- (l) ;} 
    {}
    \ifthenelse{\equal{#4}{-2}} {
    \draw[dashed,line cap=round] (ll) -- (l);
    }{}
    }
    
  \ifthenelse{\NOT \equal{#5}{-1} \AND \NOT \equal{#5}{-2}}  {
    \draw[line cap=round] (lr) -- (r); 
    \node[right] at($ (lr) ! 0.2 ! (ur) $) {\scalebox{0.8}{$#5$}};  
    \ifthenelse{\equal{#1}{1}} {
      \draw[->,>=stealth',line width=0.01pt] (lr) -- ($ (lr) ! 0.3 ! (ur) $);
          }{}    
    }{
    \ifthenelse{\equal{#5}{-1}} {
    \draw[dotted,line cap=round] (lr) -- (r) ;} 
    {}
    \ifthenelse{\equal{#5}{-2}} {
    \draw[dashed,line cap=round] (lr) -- (r);
    }{}
    }    
    
  \ifthenelse{\NOT \equal{#6}{-1} \AND \NOT \equal{#6}{-2}}  {
    \draw[line cap=round] (r) -- (l); 
    \node[above] at($ (r) ! 0.5 ! (l) $) {\scalebox{0.8}{$#6$}};  
    \ifthenelse{\equal{#1}{1}} {
      \draw[->,>=stealth',line width=0.01pt] (r) -- ($ (r) ! 0.6 ! (l) $);
          }{}    
    }{
    \ifthenelse{\equal{#6}{-1}} {
    \draw[dotted,line cap=round] (r) -- (l) ;} 
    {}
    \ifthenelse{\equal{#6}{-2}} {
    \draw[dashed,line cap=round] (r) -- (l);
    }{}
    }

   \end{tikzpicture}
  }

 \newcommand{\Xgraph}[6][0]{
  \begin{tikzpicture}[scale=0.8,baseline=-3]
    \coordinate (ul) at (-0.5,0.75);
    \coordinate (lr) at (0.5,-0.75);
    \coordinate (ll) at (-0.5,-0.75);
    \coordinate (ur) at (0.5,0.75);    
    \coordinate (u) at (0, 0.25);
    \coordinate (d) at (0,-0.25);
    
  \ifthenelse{\NOT \equal{#2}{-1} \AND \NOT \equal{#2}{-2}}  {
    \draw[line cap=round] (u) -- (ul); 
    \node[left] at($ (u) ! 0.4 ! (ul) $) {\scalebox{0.8}{$#2$}};  
    \ifthenelse{\equal{#1}{1}} {
      \draw[->,>=stealth',line width=0.01pt] (u) -- ($ (u) ! 0.6 ! (ul) $);
          }{}    
    }{
    \ifthenelse{\equal{#2}{-1}} {
    \draw[dotted,line cap=round] (u) -- (ul) ;} 
    {}
    \ifthenelse{\equal{#2}{-2}} {
    \draw[dashed,line cap=round] (u) -- (ul);
    }{}
    }    
    
  \ifthenelse{\NOT \equal{#3}{-1} \AND \NOT \equal{#3}{-2}}  {
    \draw[line cap=round] (u) -- (ur); 
    \node[right] at($ (u) ! 0.45 ! (ur) $) {\scalebox{0.8}{$#3$}};  
    \ifthenelse{\equal{#1}{1}} {
      \draw[->,>=stealth',line width=0.01pt] (u) -- ($ (u) ! 0.6 ! (ur) $);
          }{}    
    }{
    \ifthenelse{\equal{#3}{-1}} {
    \draw[dotted,line cap=round] (u) -- (ur) ;} 
    {}
    \ifthenelse{\equal{#3}{-2}} {
    \draw[dashed,line cap=round] (u) -- (ur);
    }{}
    }

  \ifthenelse{\NOT \equal{#4}{-1} \AND \NOT \equal{#4}{-2}}  {
    \draw[line cap=round] (ll) -- (d); 
    \node[left] at($ (ll) ! 0.6 ! (d) $) {\scalebox{0.8}{$#4$}};  
    \ifthenelse{\equal{#1}{1}} {
      \draw[->,>=stealth',line width=0.01pt] (ll) -- ($ (ll) ! 0.6 ! (d) $);
          }{}    
    }{
    \ifthenelse{\equal{#4}{-1}} {
    \draw[dotted,line cap=round] (ll) -- (d) ;} 
    {}
    \ifthenelse{\equal{#4}{-2}} {
    \draw[dashed,line cap=round] (ll) -- (d);
    }{}
    }
    
  \ifthenelse{\NOT \equal{#5}{-1} \AND \NOT \equal{#5}{-2}}  {
    \draw[line cap=round] (lr) -- (d); 
    \node[right] at($ (lr) ! 0.6 ! (d) $) {\scalebox{0.8}{$#5$}};  
    \ifthenelse{\equal{#1}{1}} {
      \draw[->,>=stealth',line width=0.01pt] (lr) -- ($ (lr) ! 0.6 ! (d) $);
          }{}    
    }{
    \ifthenelse{\equal{#5}{-1}} {
    \draw[dotted,line cap=round] (lr) -- (d) ;} 
    {}
    \ifthenelse{\equal{#5}{-2}} {
    \draw[dashed,line cap=round] (lr) -- (d);
    }{}
    }    
    
  \ifthenelse{\NOT \equal{#6}{-1} \AND \NOT \equal{#6}{-2}}  {
    \draw[line cap=round] (d) -- (u); 
    \node[right] at($ (d) ! 0.5 ! (u) $) {\scalebox{0.8}{$#6$}};  
    \ifthenelse{\equal{#1}{1}} {
      \draw[->,>=stealth',line width=0.01pt] (d) -- ($ (d) ! 0.6 ! (u) $);
          }{}    
    }{
    \ifthenelse{\equal{#6}{-1}} {
    \draw[dotted,line cap=round] (d) -- (u) ;} 
    {}
    \ifthenelse{\equal{#6}{-2}} {
    \draw[dashed,line cap=round] (d) -- (u);
    }{}
    }      

   \end{tikzpicture}
  }

\newcommand{\propa}[3][0]{
\begin{tikzpicture}[scale=0.6,baseline=-3]

  \ifthenelse{\NOT \equal{#2}{-1} \AND \NOT \equal{#2}{-2}}  {
    \draw[line cap=round] (0,-0.625) -- (0,0.625); 
    \ifthenelse{\equal{#3}{0}}{
    \node[right,inner sep=1, outer sep=1] at(0,0) {\scalebox{0.8}{$#2$}};}
    {\node[left,inner sep=1, outer sep=1] at(0,0) {\scalebox{0.8}{$#2$}};}  
    \ifthenelse{\equal{#1}{1}} {
      \draw[->,>=stealth',line width=0.01pt] (0,-0.625) -- (0,0.1);
          }{}  
    \ifthenelse{\equal{#1}{2}} {
      \draw[-<,>=stealth',line width=0.01pt] (0,-0.625) -- (0,0.1);
          }{}         
    }{
    \ifthenelse{\equal{#2}{-1}} {
    \draw[dotted,line cap=round] (0,-0.625) -- (0,0.625);} 
    {}
    \ifthenelse{\equal{#2}{-2}} {
    \draw[dashed,line cap=round] (0,-0.625) -- (0,0.625);
    }{}
    }      

   \end{tikzpicture}
}

\newcommand{\bubbleGraph}[6]{
\begin{tikzpicture}[scale=1,baseline=-3]
  \draw[line cap=round] (0,0.25) arc (90:270:0.15 and 0.25) ; 
  \draw [-<,>=stealth', line width=0.01pt] (0,0.25) arc (90:190:0.15 and 0.25); 
  \node[left, inner sep=2,outer sep=3] at(-0.05,0) {\scalebox{0.7}{$#3$}};

  \draw[line cap=round] (0,-0.25) arc (-90:90:0.15 and 0.25); 
  \draw [->,>=stealth', line width=0.01pt]  (0,-0.25) arc (-90:20:0.15 and 0.25); 
  \node[right] at(0.05,0) {\scalebox{0.7}{$#4$}};

  \draw[line cap=round] (0,-0.625) -- (0,-0.25); 
  \draw [->,>=stealth', line width=0.01pt] (0,-0.625) -- (0,-0.45) node[ left, inner sep=0,outer sep=2] {\scalebox{0.7}{$#1$}};

  \draw[line cap=round] (0,0.25) -- (0,0.625); 
  \draw [->,>=stealth', line width=0.01pt] (0,0.25) -- (0,0.5) node[left, inner sep=0,outer sep=2] {\scalebox{0.7}{$#2$}};

  \node[below] at(0,0.3) {\scalebox{0.6}{$#5$}};
  \node[above] at(0,-0.35) {\scalebox{0.6}{$#6$}};
  \end{tikzpicture}
  }

\newcommand{\Rxoss}[4][0]{
\begin{tikzpicture}[scale=0.6,baseline]
 \coordinate (ll) at (-0.5, -0.6);
 \coordinate (lr) at (0.5, -0.6);
 \coordinate (ul) at (-0.5,0.6);
 \coordinate (ur) at (0.5,0.6);
 \coordinate (c) at (0,0);
  
  \ifthenelse{\equal{#1}{0}}{
   \ifthenelse{\NOT \equal{#4}{-1} \AND \NOT \equal{#4}{-2}}{
    \draw[line cap=round] (lr) -- (ul); 
    \node[right] at ($ (c) ! 0.5 ! (lr) $) {\scalebox{0.5}{$#4$}};
    \ifthenelse{\equal{#2}{1}}{
      \draw[->,>=stealth',line width=0.01pt] (c) -- ($ (c) ! 0.7 ! (ul) $);
     }{}
    }{
    \ifthenelse{\equal{#4}{-1}}{
     \draw[line cap=round,dotted] (lr) -- (ul);
     }{}
    \ifthenelse{\equal{#4}{-2}}{
     \draw[line cap=round,dashed] (lr) -- (ul);
     }{}     
    }
    
\path[fill=white] (c) circle (2pt); 

   \ifthenelse{\NOT \equal{#3}{-1} \AND \NOT \equal{#3}{-2}}{
   \draw[line cap=round] (ll) -- (ur);     
   \node[left] at ($ (c) ! 0.5 ! (ll) $) {\scalebox{0.5}{$#3$}};
   \ifthenelse{\equal{#2}{1}}{
      \draw[->,>=stealth',,fill=black,line width=0.01pt] (c) -- ($ (c) ! 0.7 ! (ur) $);
     }{}
    }{
    \ifthenelse{\equal{#3}{-1}}{
     \draw[line cap=round,dotted] (ll) -- (ur);
     }{}
    \ifthenelse{\equal{#3}{-2}}{
     \draw[line cap=round,dashed] (ll) -- (ur);
     }{}     
    }
  }  
  {
   \ifthenelse{\NOT \equal{#3}{-1} \AND \NOT \equal{#3}{-2}}{
   \draw[line cap=round] (ll) -- (ur);     
   \node[left] at ($ (c) ! 0.5 ! (ll) $) {\scalebox{0.5}{$#3$}};
   \ifthenelse{\equal{#2}{1}}{
      \draw[->,>=stealth',line width=0.01pt] (c) -- ($ (c) ! 0.7 ! (ur) $);
     }{}
    }{
    \ifthenelse{\equal{#3}{-1}}{
     \draw[line cap=round,dotted] (ll) -- (ur);
     }{}
    \ifthenelse{\equal{#3}{-2}}{
     \draw[line cap=round,dashed] (ll) -- (ur);
     }{}     
    }  

\path[fill=white] (c) circle (2pt); 
  
   \ifthenelse{\NOT \equal{#4}{-1} \AND \NOT \equal{#4}{-2}}{
    \draw[line cap=round] (lr) -- (ul);
    \node[right] at ($ (c) ! 0.5 ! (lr) $) {\scalebox{0.5}{$#4$}};
    \ifthenelse{\equal{#2}{1}}{
      \draw[->,>=stealth',fill=black,line width=0.01pt] (c) -- ($ (c) ! 0.7 ! (ul) $);
     }{}
    }{
    \ifthenelse{\equal{#4}{-1}}{
     \draw[line cap=round,dotted] (lr) -- (ul);
     }{}
    \ifthenelse{\equal{#4}{-2}}{
     \draw[line cap=round,dashed] (lr) -- (ul));
     }{}     
    }
  }
 
\end{tikzpicture}
}

\newcommand{\Rxosses}[5][0]{
 \begin{tikzpicture}
  \ifnum #5=1 {\node {\Rxoss[#1]{#2}{#3}{#4}};}
  \else {
   \foreach \x in {1,...,#5} {
    \ifnum \x=1 {
     \ifthenelse{\equal{#3}{-1} \OR \equal{#3}{-2}}{
      \ifthenelse{\equal{#4}{-1} \OR \equal{#4}{-2}}{
      \node[inner sep=-0.2pt,outer sep=0pt] {\Rxoss[#1]{#2}{#3}{#4}};
     }
     {\node[inner sep=-0.2pt,outer sep=0pt] {\Rxoss[#1]{#2}{#3}{}};}
     }
     {\ifthenelse{\equal{#4}{-1}\OR\equal{#4}{-2}}{
           \node[inner sep=-0.2pt,outer sep=0pt] {\Rxoss[#1]{#2}{}{#4}};
          }
          {\node[inner sep=-0.2pt,outer sep=0pt] {\Rxoss[#1]{#2}{}{}};}}
    }
    \else {
     \ifthenelse{\equal{#3}{-1} \OR \equal{#3}{-2}}{
      \ifthenelse{\equal{#4}{-1} \OR \equal{#4}{-2}}{
       \ifodd\x {
        \node[inner sep=-0.2pt,outer sep=0pt,anchor=south west] at (current bounding box.north west) {\Rxoss[#1]{#2}{#3}{#4}};}
       \else{
        \node[inner sep=-0.2pt,outer sep=0pt,anchor=south west] at (current bounding box.north west) {\Rxoss[#1]{#2}{#4}{#3}};
       }
       \fi
     }
     {
      \ifodd\x {
        \node[inner sep=-0.2pt,outer sep=0pt,anchor=south west] at (current bounding box.north west) {\Rxoss[#1]{#2}{#3}{}};}
       \else{
        \node[inner sep=-0.2pt,outer sep=0pt,anchor=south west] at (current bounding box.north west) {\Rxoss[#1]{#2}{}{#3}};
       }
       \fi
     }
     }
     {\ifthenelse{\equal{#4}{-1} \OR \equal{#4}{-2}}{
      \ifodd\x {
        \node[inner sep=-0.2pt,outer sep=0pt,anchor=south west] at (current bounding box.north west) {\Rxoss[#1]{#2}{}{#4}};}
       \else{
        \node[inner sep=-0.2pt,outer sep=0pt,anchor=south west] at (current bounding box.north west) {\Rxoss[#1]{#2}{#4}{}};    }
       \fi
          }{
       \node[inner sep=-0.2pt,outer sep=0pt,anchor=south west] at (current bounding box.north west) {\Rxoss[#1]{#2}{}{}};}
      } 
     \ifnum \x=#5 {              
      \ifthenelse{\NOT \equal{#3}{-1} \AND \NOT \equal{#3}{-2}}{
      \node at (-0.25,-0.2) {\scalebox{0.5}{$#3$}};}{}
      \ifthenelse{\NOT \equal{#4}{-1} \AND \NOT \equal{#4}{-2}}{
      \node at (0.4,-0.2) {\scalebox{0.5}{$#4$}};}{}
      }\fi
    }
    \fi
   } 
  }   
  \fi
 \end{tikzpicture}
}

\newcommand{\Yupst}[1]{
\begin{tikzpicture}[scale=0.5,baseline={(current bounding box.center)}]
    \draw[line cap=round]   (0,0) -- (150:1) -- +(0,0.75) ;
    \draw[->,>=stealth',line width=0.01pt]  (0,0) -- (150:0.5) ;    
    \draw[line cap=round]  (0,0) -- (30:1) -- +(0,0.75) ;
    \draw[->,>=stealth',line width=0.01pt]  (0,0) -- (30:0.5) ;
    \draw[line cap=round] (0,0) -- (0,-1);
    \draw[->,>=stealth',line width=0.01pt]  (0,-1) --  (0,-0.4);  
    \draw[line cap=round,dashed] (0,-0.7) -- +(30:1) -- (30:1);
    \node[above] at(150:0.5) {\scalebox{0.7}{$r$}};
    \node[above] at(30:0.5) {\scalebox{0.7}{$s$}};  
    \node[left] at(0,-0.5) {\scalebox{0.7}{$t$}};   
    \node[above] at(0,-0.1) {\scalebox{0.7}{$#1$}}; 
\end{tikzpicture}
}

\newcommand{\Yuprt}[1]{
\begin{tikzpicture}[scale=0.5,baseline={(current bounding box.center)}]
    \draw[line cap=round]   (0,0) -- (150:1) -- +(0,0.75) ;
    \draw[->,>=stealth',line width=0.01pt]  (0,0) -- (150:0.5) ;    
    \draw[line cap=round]  (0,0) -- (30:1) -- +(0,0.75) ;
    \draw[->,>=stealth',line width=0.01pt]  (0,0) -- (30:0.5) ;
    \draw[line cap=round] (0,0) -- (0,-1);
    \draw[->,>=stealth',line width=0.01pt]  (0,-1) --  (0,-0.4);  
    \draw[line cap=round,dashed] (0,-0.7) -- +(150:1) -- (150:1);
    \node[above] at(150:0.5) {\scalebox{0.7}{$r$}};
    \node[above] at(30:0.5) {\scalebox{0.7}{$s$}};  
    \node[right] at(0,-0.5) {\scalebox{0.7}{$t$}}; 
    \node[above] at(0,-0.1) {\scalebox{0.7}{$#1$}};        
\end{tikzpicture}
}

\newcommand{\Yupsr}[1]{
\begin{tikzpicture}[scale=0.5,baseline={(current bounding box.center)}]
    \draw[line cap=round]   (0,0) -- (150:1) -- +(0,0.75);
    \draw[->,>=stealth',line width=0.01pt]  (0,0) -- (150:0.5) ;    
    \draw[line cap=round]  (0,0) -- (30:1) -- +(0,0.75);
    \draw[->,>=stealth',line width=0.01pt]  (0,0) -- (30:0.5) ;
    \draw[line cap=round] (0,0) -- (0,-1);
    \draw[->,>=stealth',line width=0.01pt]  (0,-1) --  (0,-0.4);  
    \draw[line cap=round,dashed] (30:1) -- +(165:1.78);
    \node[below] at(150:0.5) {\scalebox{0.7}{$r$}};
    \node[below] at(30:0.5) {\scalebox{0.7}{$s$}};  
    \node[right] at(0,-0.6) {\scalebox{0.7}{$t$}}; 
    \node[above] at(0,-0.1) {\scalebox{0.7}{$#1$}};        
\end{tikzpicture}
}

\newcommand{\VRxoss}[6][0]{
\begin{tikzpicture}[scale=0.8,baseline]
 \coordinate (c) at (0,-1);
 \coordinate (a) at (0.7085,0.7085);
 \coordinate (b) at (-0.7085,0.7085);
 \coordinate (v) at (0,-0.25); 
 \coordinate (x) at (0,0.25);  
 \coordinate (ctra) at (-0.25,-0.25);
 \coordinate (ctraa) at (-0.35,0.2);
 \coordinate (ctrb) at (0.25,-0.25);
 \coordinate (ctrbb) at (0.35,0.2);
 
  \ifthenelse{\equal{#1}{0}}{
   \ifthenelse{\NOT \equal{#3}{-1} \AND \NOT \equal{#3}{-2}}{
   \draw[line cap=round] (v) .. controls (ctrb) and (ctrbb) .. (b);     
   \node[left,inner sep=0] at ($ (ctrb) ! 0.75 ! (b) $) {\scalebox{0.7}{$#3$}};
   \ifthenelse{\equal{#2}{1}}{
      \draw[line cap=round,line width=0.01pt, arrow data={0.8}{stealth'}] (v) .. controls (ctrb) and (ctrbb) .. (b);
     }{}
    }{
    \ifthenelse{\equal{#3}{-1}}{
     \draw[line cap=round,dotted] (v) .. controls (ctrb) and (ctrbb) .. (b); 
     }{}
    \ifthenelse{\equal{#3}{-2}}{
     \draw[line cap=round,dashed] (v) .. controls (ctrb) and (ctrbb) .. (b); 
     }{}     
    }

\path[fill=white] (x) circle (2pt); 
  
   \ifthenelse{\NOT \equal{#4}{-1} \AND \NOT \equal{#4}{-2}}{
    \draw[line cap=round] (v) .. controls (ctra) and (ctraa) .. (a); 
    \node[right,inner sep=0] at ($ (ctra) ! 0.75 ! (a) $) {\scalebox{0.7}{$#4$}};
    \ifthenelse{\equal{#2}{1}}{
      \begin{scope}[fill=black] 
      \draw[line cap=round,line width=0.01pt, arrow data={0.8}{stealth'}] (v) .. controls (ctra) and (ctraa) .. (a); \end{scope}
     }{}
    }{
    \ifthenelse{\equal{#4}{-1}}{
     \draw[line cap=round,dotted] (v) .. controls (ctra) and (ctraa) .. (a);
     }{}
    \ifthenelse{\equal{#4}{-2}}{
     \draw[line cap=round,dashed] (v) .. controls (ctra) and (ctraa) .. (a);
     }{}     
    }
  }  
  {
   \ifthenelse{\NOT \equal{#4}{-1} \AND \NOT \equal{#4}{-2}}{
    \draw[line cap=round] (v) .. controls (ctra) and (ctraa) .. (a); 
    \node[right,inner sep=0] at ($ (ctra) ! 0.75 ! (a) $) {\scalebox{0.7}{$#4$}};
    \ifthenelse{\equal{#2}{1}}{
      \draw[line cap=round,line width=0.01pt, arrow data={0.8}{stealth'}] (v) .. controls (ctra) and (ctraa) .. (a);
     }{}
    }{
    \ifthenelse{\equal{#4}{-1}}{
     \draw[line cap=round,dotted] (v) .. controls (ctra) and (ctraa) .. (a);
     }{}
    \ifthenelse{\equal{#4}{-2}}{
     \draw[line cap=round,dashed] (v) .. controls (ctra) and (ctraa) .. (a);
     }{}     
    }
    
\path[fill=white] (x) circle (2pt); 

   \ifthenelse{\NOT \equal{#3}{-1} \AND \NOT \equal{#3}{-2}}{
   \draw[line cap=round] (v) .. controls (ctrb) and (ctrbb) .. (b);     
   \node[left, inner sep=0] at ($ (ctrb) ! 0.75 ! (b) $) {\scalebox{0.7}{$#3$}};
   \ifthenelse{\equal{#2}{1}}{
      \begin{scope}[fill=black] 
      \draw[line cap=round,line width=0.01pt, arrow data={0.8}{stealth'}] (v) .. controls (ctrb) and (ctrbb) .. (b);\end{scope}
     }{}
    }{
    \ifthenelse{\equal{#3}{-1}}{
     \draw[line cap=round,dotted] (v) .. controls (ctrb) and (ctrbb) .. (b); 
     }{}
    \ifthenelse{\equal{#3}{-2}}{
     \draw[line cap=round,dashed] (v) .. controls (ctrb) and (ctrbb) .. (b); 
     }{}     
    }
  }   
  
\ifthenelse{\NOT \equal{#5}{-1} \AND \NOT \equal{#5}{-2}}{
\draw[line cap=round] (c) -- (v);
\node[right,inner sep=2] at ($ (c) ! 0.5 ! (v) $) {\scalebox{0.7}{$#5$}};
\ifthenelse{\equal{#2}{1}}{
   \begin{scope}[fill=black]
   \draw[line cap=round,line width=0.01pt, arrow data={0.9}{stealth'}] (c) -- ($ (c) ! 0.6 ! (v) $); \end{scope}
     }{}
}{
    \ifthenelse{\equal{#5}{-1}}{
     \draw[line cap=round,dotted] (c) -- (v);
     }{}
    \ifthenelse{\equal{#5}{-2}}{
     \draw[line cap=round,dashed] (c) -- (v);
     }{}     
}
 
\node[inner sep=2.5,right] at (v) {\scalebox{0.7}{$#6$}};
\end{tikzpicture}
}

\usepackage{varioref}

\begin{document}

\title{A Unified Framework of Topological Phases with Symmetry}

\author{Yuxiang Gu}
\email{yuxianggu@gmail.com}
\noaffiliation
\author{Ling-Yan Hung}
\email{jhung@perimeterinstitute.ca}
\affiliation{Department of Physics, Harvard University, Cambridge MA 02138, USA}
\author{Yidun Wan}
\email{ywan@perimeterinstitute.ca}
\affiliation{Perimeter Institute for Theoretical Physics, Waterloo, ON N2L 2Y5, Canada}

\date{\today}

\begin{abstract}
In topological phases in $2+1$ dimensions, anyons fall into representations of quantum group symmetries. As proposed in our work (arXiv:1308.4673), physics of a symmetry enriched phase can be extracted by the Mathematics of (hidden) quantum group symmetry breaking of a ``parent phase''. This offers a unified framework and classification of the symmetry enriched (topological) phases, including symmetry protected trivial phases as well. In this paper, we extend our investigation to the case where the ``parent" phases are non-Abelian topological phases. We show explicitly how one can obtain the topological data and symmetry transformations of the symmetry enriched phases from that of the ``parent" non-Abelian phase. Two examples are computed: (1) the \II phase breaks into the $\Z_2$ toric code with $\Z_2$ global symmetry; (2) the $SU(2)_8$ phase breaks into the chiral Fibonacci $\x$ Fibonacci phase with a $\Z_2$ symmetry, a first non-Abelian example of symmetry enriched topological phase beyond the gauge theory construction. \end{abstract}
\pacs{11.15.-q, 71.10.-w, 05.30.Pr, 71.10.Hf, 02.10.Kn, 02.20.Uw}
\maketitle

\section{Introduction}\label{sec:intro}
Recently, there is tremendous progress \cite{Chen2011f,Senthil2012a,Vishwanath2012,Essin2012,Mesaros2011,
Hung2012a,Wang2012,Lu2013,Levin,HungWan2013a,Metlitski2013,Ye2013,Sule2013,Xu2013,Huang2013,Essin2014,Nussinov2009} in understanding gapped phases with symmetries, namely the symmetry enriched phases (SEPs), including symmetry protected trivial (SPT) and symmetry enriched topological (SET) phases. Understanding these phases is a fundamental question in condensed matter physics which also promises potential applications, such as in topological quantum computation. One central question is the classification of possible phases allowed in principle in nature. In the case of SPT phases, the question has been very much answered in the seminal paper \cite{Chen2011f} in arbitrary dimensions, at least for symmetries not involving time reversal using the machinery of group cohomology theory. The same cannot be said about SET's, where a systematic classification is as yet lacking. One important development is the observation of a connection between SPT phases and Dijkgraaf-Witten (DW) lattice gauge theories, which are classified by the same group cohomology classes\cite{Levin2012,Hung2012}. By gauging the global symmetry of an SPT, one obtains a corresponding DW lattice gauge theory, which is a topological theory characterizing a topological phase. Such a connection has inspired the twisted quantum double (TQD) model of topological phases\cite{Hu2012a} and a partial classification of SET's\cite{Mesaros2011,Hung2012a}, which however remains incomplete, particularly when the topological phases involved before and after gauging are beyond the realms of group cohomology.   
To tackle the question in 2+1 dimensions, we propose a generalized framework\cite{Hung2013} that encodes all the topological data of SEPs in intrinsic topological orders describable by tensor categories, beyond those connected to DW theories.

Let us summarize our main results below, which will be followed by more detailed explanation.
\begin{itemize} 
\item We unify SPT and SET phases by embedding them in phases with intrinsic topological order. SEPs are then classified by topological phases and the embedding structures. Our results are also closely related to properties of defects in a phase.
\item We can readily predict allowed fractional global charges in different sectors in a topological phase with symmetry.
\item We explicity show how the doubled Ising model gives rise to the $\Z_2$ spin liquid with a $\Z_2$ symmetry that induces the electric-magnetic duality.
\item For the first time, we construct the chiral Fibonnaci $\times$ Fibonnaci phase with a particular $\Z_2$ symmetry. 
\end{itemize}

The Physics of SEPs is recovered from the parent intrinsic topological order via the idea of \emph{anyon condensation}, which breaks a quantum group symmetry\cite{Bais2002,Bais2003}, the hidden symmetry of a general topological phase. We note that however, no physical condensation actually occurs in our consideration. This picture reproduces all the 2+1 $d$ SET phases thus far constructed via the projective symmetry group (PSG), and opens up an  avenue to a vast number of phases hitherto unknown. It is tempting to conjecture that this is the complete framework that describes all SET phases in 2+1 $d$, which is the crucial element towards an ultimate classification. To better convey the idea just described in this paragraph, let us first briefly introduce the concepts of SEPs and hidden quantum group symmetry.

Global symmetry charges of excitations in SEP often fractionalize. Global symmetry might also
act by mapping one topological sector to the other, e.g., symmetry between charge and flux. The global symmetry must not act arbitrarily on the topological sectors, as pointed out in Ref.\cite{Essin2012,HungWan2013a}; rather, its action has to be compatible with the fusion rules of the topological sectors, in the sense that the global symmetry must transform physical sectors---bosons---linearly, in particular, its identity element ought to transform the bosons trivially; however,  since the fusion rules dictate which and how topological sectors fuse to bosons, this compatibility condition leaves space for the identity to act on unphysical sectors non trivially as long as it preserves the bosons, endowing the unphysical topological sectors fractionalized charges associated with the global symmetry. The nontrivial actions of the global symmetry other than that of the identity element may be able to exchange some topological sectors, leading to the charge-flux symmetry transformation mentioned above, and even more exotic behavior of the topological sectors\cite{Lu2013,HungWan2013a}.

In gauge theories, particles are characterized by their gauge charges that are labeled by the irreducible representations of the relevant gauge group. Although topological phases are usually described by effective (continuous or discrete) gauge theories with a gauge group $G,$ one soon finds that the  $G$ itself is not able to fully specify the topological sectors or anyons in the topological phase. For example, in certain simple Abelian topological phases, a generic anyon carries not only a charge quanta but also a $G$-valued flux quanta that is described by a non-gauge symmetry dual to the gauge symmetry $G$. It turns out that these anyons are representations of the quantum double $D[G]=F[G]\otimes\mathbb{C}[G]$ of $G$, where $F[G]$ is the space of functions over $G$ and $\mathbb{C}[G]$ the group algebra of $G$. $D[G]$ is a special and simple case of what is known as a quantum group, which is in fact an algebra rather than a group. More generally, as in the deconfined and low-energy limit, local interactions of the anyons in a topological phase comprise an algebra---a fusion algebra---that is the fusion ring of the representations of a quantum group, e.g., as to be seen later in the paper, the deformed universal enveloping algebra of a Lie algebra, while the long range or non-local interactions are of the Aharonov-Bohm type. One then often recognizes the relevant quantum group as a \textit{hidden quantum group symmetry} of the topologically ordered system, as a generalization of gauge symmetry, which is also hidden but instead described by a group.

In section \ref{sec:genScheme}, we review our framework of SEP construction via quantum group symmetry breaking first introduced in our previous work\cite{Hung2013}. We will then describe how the framework of anyon condensation developed in Ref.\cite{Eliens2013} can be adopted for our purpose, 
giving a precise prescription of how topological data is retrieved from the parent phase  involving more general non-Abelian topological orders. In particular, we will describe braiding properties between ``unconfined'' and ``confined'' anyons in the broken phase and how they store global symmetry transformation properties of the SEP. 
In Sections \ref{sec:Ising2break} and \ref{sec:su2_8break}, we will demonstrate our ideas in concrete examples. We will begin with a warm up example where the \II theory is considered, which encodes the $\Z_2$ symmetric toric code model\cite{Lu2013,HungWan2013a}. This is an extension of the discussion in Ref.\cite{Hung2013} which we now reformulate in a more systematic footing. In section \ref{sec:su2_8break} we will discuss a more interesting example, in which we demonstrate how a $\Z_2$ symmetric Fibonacci $\times$ Fibonacci theory is embedded in a $SU(2)_8$ topological order. To our knowledge , this is the first example of a non-Abelian topological order enriched by a global symmetry beyond the DW lattice gauge theory construction. 

\section{The general scheme of quantum group symmetry breaking}\label{sec:genScheme}
In this section, we elaborate on the general scheme of quantum group symmetry breaking and in particular how a global symmetry is generated on the residual phase via this breaking. 

To understand quantum group breaking, we begin with a special case---the PSG. Consider an SET phase $S$, a deconfined gauge theory with gauge group  $N_g$ and a global symmetry $G_s$; for later convenience, we call the topological order described by the gauge group $N_g$ the {\bf bare phase} of $S$. The phase $S$ is specified by assigning compatible $G_s$ charges, potentially fractionalized, to the $N_g$ gauge charges (chargeons), fluxes (fluxons) and dyons, which are generally anyons. A PSG group $G$ encodes such an assignment, such that $N_g$ is its normal subgroup, and that $G_s = G/N_g$.  The embedding of $N_g$ in $G$ determines the non-trivial $G_s$ representations the anyons fall into.

An alternative way exists, however, to encode the quotient group structure, by choosing a ``charge condensate'' $\vac$ such that $N_g$ is the invariant subgroup of $G$ that keeps $\vac$ invariant. This is familiar in the Higgs mechanism, where the quotient $G_s= G/N_g$ describes the moduli space, and in the PSG construction $G_s$ is interpreted as the global symmetry. The condensate $\vac$ by definition transforms under $G_s$ and thus carries a $G_s$ charge, but it carries no $N_g$ charge and thus automatically belongs to the topologically trivial sector of the $N_g$ gauge theory. To emphasize the distinction from a physical condensate, let us call $\vac$ the {\bf charged vacuum}.

In the context of a topological order, the chargeons and fluxons of the deconfined gauge theory are on an equal footing. The ``anyon condensate'' perspective of the PSG construction therefore immediately suggests a generalization by considering more general $\vac$.

We thus need to go beyond gauge symmetry breaking. It turns out that quantum group symmetry breaking, which naturally treats chargeons, fluxons and dyons equally, can describe general anyon condensates. In fact, for a general topological phase not necessarily related to a gauge theory, each type of anyon can be understood as the representation of some hidden Hopf symmetry algebra associated to the phase.  Hopf symmetry breaking has been discussed in the context of anyon condensation\cite{Bais2002,Bais2003,Bais2009,Bais2009a,Bais2012,You2013}, where a macroscopic condensate of anyons appears, driving a phase transition that takes one topological phase described by a Hopf algebra/quantum group $\mathcal{A}$ to another phase described by $\mathcal{U}$. Here, the \emph{same} Mathematics encodes the physics of SEPs. A topological phase with a quantum group symmetry $\A$, which plays the role of the PSG, encodes an SEP phase whose bare phase has a smaller quantum group symmetry $\mathcal{U}$. Let us summarize how this is to be done in the following three steps. 

\subsection{Key steps}
\textbf{Step 1: The charged vacuum $\vac$}.
As discussed above, we first choose a $\vac$ in the representations of $\mathcal{A}$. In the PSG construction $\vac$ is just a state in a representation of the PSG $G$. The topological order in the SEP $S$ is then related to the subalgebra of $\A$, denoted $\T$, that keeps $\vac$ ``\emph{invariant}'', a subtle concept that has to be defined carefully in the context of an algebra,  and such an approach has indeed been taken in Ref\cite{Bais2003}.  This approach is well defined in Hopf algebras such as the twisted quantum double (TQD)\cite{Hu2012a,Mesaros2011}, where the internal states of all representations are physical and appear in the physical Hilbert space. In these cases however, the quantum dimension of the anyon (representation) containing the "charged vacuum" is necessarily an integer. All Abelian topological phases can be described by TQDs,  and we have studied the SEPs due to the breaking of Abelian topological phases in Ref.\cite{Hung2013}.
 
More generally, anyons have fractional or even irrational quantum dimensions, e.g., the anyons described by quantum groups $U_q(sl_2)$ for $q$ a root of unity; nonetheless, it is still possible to define $\vac$ as hidden in a representation that labels an anyon $c$ of $\A$\cite{Bais2009a}, although one does not solve for the subalgebra $\T$ explicitly. The allowed quantum dimension for $\vac$ must be one, as expected of a physical excitation, and the \emph{topological spin} $h_{\vac}$ of $\vac$ should be integers, such that the self statistics $\theta_{\vac}=\exp{\ii 2\pi h_{\vac}}$ of $\vac$ should be unity. There are two possible types of condensates. The first type is \textit{simple current} condensate, in the sense that the anyon $\gamma$ containing the charged vacuum has $d_\gamma=1$ and surely $h_\gamma\in\Z$. In this case, we can identify our notion of $\vac$ with $c$ for simplicity. A subtlety exists however, that is, if $\gamma\x \gamma=\gamma'$, an anyon distinct from $\gamma$, $\gamma'$ must also condense if $\gamma$ condenses. In other words, more than one topological sectors of $\A$ may actually condense. In this paper, we shall restrict our investigation to the cases where only one sector $\gamma$ may condense, implying that $\gamma\x \gamma=1$, the trivial sector. But many results to be shown apply beyond this restriction.

Anyon condensate in $\A$ can make certain sectors of $\A$ break into irreducible representations of a subalgebra $\T\subset A$. Namely, 
\be \label{eq:restrict}
a \rightarrow \sum\nolimits_i n^{a}_i  a_i,
\ee
where $n^{a}_i$ is the multiplicity of the representation $a_i$ of $\mathcal{T}$ in the decomposition. Note that the decomposition conserves the quantum dimension, $d_{a} = \sum\nolimits_j n^{a}_i d_{a_i}$. In particular for a simple current condensate $\gamma$, a topological sector of $\A$ splits as in Eq. \eqref{eq:restrict} if and only if $a \x \gamma=a$ and $d_a>1$. One may understand this feature in three ways. First, the condensate $\gamma$ renders $a\x \gamma$ and $a$ with different global charges in the broken theory $\T$; hence, $a$ must split into two $\T$ sectors that carry the two different global charges. In fact, if there are more than one condensates and some $\A$ sector $a$ is invariant under fusion with all these condensates, $a$ must split into multiple $\T$ sectors. Second, from the unitary tensor category point of view, $a\x \gamma=a$ implies that there exists a nontrivial isomorphism on the object $a$, which further implies that when the simple object $\gamma$ condenses and is hence removed from the category $\A$, $a$ should be a nonsimple object of the resultant category $\T$, thus $a$ should take the form in Eq. \eqref{eq:restrict}. Third, since $\gamma\x \gamma=1$ and $\gamma\x a=a$, $a\x \bar{a}$ must contain two copies of the vacuum or trivial sector $1$, had $a$ not split into two sectors of $\T$, $a$ and its anti-anyon $\bar a$ would have two different ways of annihilating into the vacuum, which is forbidden by unitarity. Hence, $a$ has to split. 

The second type of anyon condensate is nonsimple current condensate, which is a partial condensate in the following sense. A \textit{nonsimple} current is a sector $\gamma$ with $d_\gamma>1$ and the fusion product of $\gamma$ with some sectors $a$ produces more than one sectors. According to Ref. \cite{Bais2009a}, if we enforce a nonsimple current $\gamma$ with $h_\gamma\in\Z$ to condense, $\gamma$ will not condense completely because $d_\gamma>1$; rather, it will split as $\gamma\rightarrow n^\gamma_11+\sum_{i>1}n^\gamma_i \gamma_i$, where $n^\gamma_1\geq 1$, into $n^\gamma_1$ copies of the vacuum of the condensed phase and other sectors depending on $d_\gamma$. In other words, the charged vacuum $\vac$ can be identified with one or more of those states that are indistinguishable from the vacuum of the condensed phase, as if $\gamma$ has an internal Hilbert space of dimension $d_\gamma$. Because the nonsimple current condensates are rather complicated and there lacks a general rule determining how topological sectors of $\A$ should split, we shall restrict our discussion hereafter to simple current condensates only. 

For similar reasons, we also restrict to the cases where $\A$'s fusion algebra is multiplicity free, namely, $N^c_{ab}= 1$ if $a$ and $b$ can fuse to $c$, otherwise $N^c_{ab}=0$.

Prior discussions of SEP phases describable by gauge theories enriched by global symmetries\cite{Hung2012a} call $\vac$  a ``pure ($G_s$) charge', which has bosonic self statistics and trivial mutual statistics with all other electrically gauge charged particles,  in accord with our discussion.

\textbf{Step 2: Identifying the ``bare phase'' and the confined sector}. One might expect the ``bare phase'' to be described by the ``invariant subalgebra'' $\T$ above, but this is not true\footnote{We assume the quasi-triangular Hopf algebra $\A$ is modular, then there is always a confined sector in $\T$, such that $\U$ is smaller than $\T$}. According to the literature on anyon condensation, an anyon condensate always divides the anyons into two distinct sets: the confined and unconfined sectors. The former consists of all the anyons that are \emph{non-local} with respect to the condensate, whereas the latter consists of those \textit{local} with respect to the condensate. Two excitations are mutually local if they have trivial mutual statistics. An anyon mutually non-local with the condensate would pull a string in the medium, and its creation or isolation is thus energetically expensive, i.e., they are indeed \emph{confined}. This concept is again applicable here in determining the ``bare phase''.  Recall in Step 1 that $\vac$ is taken as the charged vacuum; therefore, barring the $G_s$ charge it carries, as far as the ``bare phase'' is concerned, it behaves exactly as the trivial sector, which is necessarily mutually local with \emph{all} other anyons of the ``bare phase''. As such, the anyons in the confined sector cannot be part of the ``bare phase''.  Discussed in Ref\cite{Lu2013,Hung2012a,Wen2013a}, confined particles are precisely those identified as the ``twist'' particles.  As a result, quite generally, the  Hopf algebra $\U$ that characterizes the ``bare phase'' differ from the ``invariant subalgebra'' $\T$ introduced above, as some of $\T$'s anyon representations are excluded from $\U$.

To recover $\U$, we should first isolate the confined sector. For this we need the fusion algebra of the decomposed anyons $a_i$ and $b_j$, which we can deduce from two facts: 1) the fusion rules commute with (\ref{eq:restrict}), namely
\be \label{eq:commute}
a\otimes b = c \to (\sum\nolimits_{i,j} n^a_i n^b_j a_i \otimes b_j )= \sum\nolimits_k n^c_k c_k ,
\ee
which implies the conservation of quantum dimension, and 2) the fusion algebra is associative. Consider in particular the fusion of the decomposed pieces $a_i$ and $\vac$. Since $\vac$ is topologically trivial in $\U$,  any two anyons belong to the same topological sector in $\U$ if they are related by fusion with $\vac$. But if two supposedly identified anyons $a_i$ and $b_j$ descend from anyons $a$ and $b$ in $\A$ with $\theta_{a} \neq \theta_{b}$, then as anyons in $\U,$ they cease to have well defined topological spins; hence, they must be confined. According to Ref\cite{Bais2009a}, this is the necessary and sufficient condition for determining the confined anyons.

With all confined anyons excluded and all anyons related by $\vac$ identified, what remains is the unconfined sector consisting of anyons with well defined topological spins and fusion rules between them, leading to a unique $\U$. Clearly to describe SET phases, the resultant $\U$ is non-trivial with multiple topological sectors, whereas a trivial $\U$ encodes the physics of SPT phases. We note that in practice there are further consistency conditions that has to be taken into account to obtain the complete fusion algebra, as detailed in Ref.\cite{Bais2009a}.

Below we note down some important general results for future convenience in the paper. We shall first prove these results Mathematically, and postpone a physical motivation after we explain Step 3. 
\begin{proposition}\label{prop:split}
If a $\T$ sector $t$ has only one lift $a$, there must exist at least another $T$ sector $t'$ that has the same lift. In other words, $t$ and $t'$ are both restrictions of the lift $a$.
\end{proposition}
The proof is simple. Consider single simple current condensate only, suppose $u$ has no other lifts but the only one $a$, then the fusion of $a$ and the condensate $\gamma$ $a\x \gamma=a$ must hold. As elaborated in Step 1, $a$ has t split into two unconfined sectors. This proof clearly generalizes to the case with multiple single current condensates.
\begin{proposition}\label{prop:equalLifts}
All $\T$ sectors that have multiple lifts should have the same number of lifts.
\end{proposition}
Also consider single simple current condensate $\gamma$, suppose a $\T$ sector $t$ has three different lifts $\{a,a',a''\}$ and suppose $a'=a\x \gamma$, then we have either $a''=a\x \gamma$ or $a''=a'\x\gamma$. Since $\gamma\x \gamma$=1, the former implies $a''=a'$ while the latter implies $a''=a$. Thus, any $\T$ sector can have at most two distinct lifts, supporting the proposition. It is easy to generalize this derivation to the case with $n$ condensates, where one can show that any $\T$ sector that has more than one lifts must have exactly $n$ lifts. This completes the proof of the proposition. A natural corollary is that the number of lifts is equal to the number of the lifts of $\T$'s vacuum. 

\textbf{Step 3: Identifying the global symmetry group action and the pure braids}. The story of anyon condensation ends, but our journey to SEPs continues: we need to find out the global symmetry $G_s$ and how it transforms the anyons in $\U$. The key lies in the confined sector. Let us  propose a rule:

\centerline{\begin{boxedminipage}[c]{0.98\columnwidth}One can extract the group action of an element $g_i$ of $G_s$ by winding the physical system (a set of particles in $\U$) around a \textit{spectator}, a specific confined particle $c_i$ corresponding to $g_i$. The fusion of the spectators modulo the fusion of the unconfined sector yields the structure of $G_s$. That is, the confined sector generates the group action of the non-trivial elements of $G_s$; the unconfined particles obtainable from fusing the confined particles generate the action of the identity, and are thus related to global charge fractionalization.\end{boxedminipage}}

The gauging procedure in Ref\cite{Levin2012,Hung2012,Hung2012a,Lu2013} inspires the above proposal. When the global symmetry is gauged, the new theory has an expanded spectrum containing excitations carrying ``magnetic flux'' of the global symmetry $G_s$. In those cases studied, the Aharanov-Bohm phase acquired by a  $G_s$ charge moving about a $G_s$ flux coincides with a particular group action of $G_s$, depending on the flux particle involved.

This can be contrasted with Ref\cite{Essin2012} which gives a consistency condition of the $G_s$ action, that the identity element of $G_s$ must act on individual anyons in a way such that its aggregate action on a system of anyons fusing to a physical state is trivial. This condition strongly points to identifying confined particles as ``generators of $G_s$'', up to their fusion to unconfined particles. This also complies with our intuition from group symmetry breaking, where $G_s$ is certain quotient of $\A$ by $\U$.

The above scheme of generating the global symmetry enables one to predict which topological sectors in a bare phase can acquire fractionalized global symmetry charges, and find out all such possible fractionalizations. As shown in Ref.\cite{HungWan2013a}, global symmetry fractionalization depends on the bare phase only and is determined by the action of the identity element of the global symmetry group; therefore, in the boxed rule above, one can take an unconfined sector as the spectator, whose monodromy with the $\U$ sectors will generate the action of the identity element of the global symmetry on $\U$. As an example, in our previous work\cite{Hung2013}, we could eaisly find that for the Ising type topological phase as the bare phase, only the sector $\sigma$ may carry a $-1$ global symmetry charge. Without using our quantum group breaking approach, it is rather involved to make such predictions\cite{Yao2010}\footnote{Yong-Shi Wu, \textit{Recent Progress in 2d Exactly Solvable Discrete Models of Topological Phases}, and Hong Yao, \textit{Topological phases enriched by non-spatial and/or spatial symmetries}, 2013 IASTU Summer Forum on the Interplay of Symmetry and Topology in Condensed Matter Physics, Institute for Advanced Study, Tsinghua University, Beijing, China, July, 2013}.
\subsection{$\T$ and $\U$ from $\A$ in detail}
To derive the precise topological content of $\T$ and $\U$ phases from $\A$, we first need to establish a relation between the Hilbert space of the broken phase and that of the original $\A$ phase. In our previous work\cite{Hung2013}, we made an attempt of mapping the Hilbert space of the broken theory to that of the parent theory, with which we demonstrated how we can define the braiding and monodromy between an unconfined particle and a confined one. This mapping of the Hilbert space is defined by the vertex lifting coefficients (VLCs) devised by Elie\"ns et al \cite{Eliens2013}, which is almost identical to our attempt apart from a treatment of the confined sectors. We now briefly review this framework, with adaptations to our purposes, in particular in extracting the group action of the global symmetry on the unconfined phase $\U$. We shall adopt a celebrated diagramatics---spacetime diagrams---of anyon models, whose underlying mathematial framework is the unitary braided fusion category theory, to illustrate the key ideas and procedures. As one can find nice reviews of the diagramatics in various resources, e.g., Ref.\cite{Eliens2013, Kitaev2006}, we shall not include the rudiments here. Rather, we make a few remarks on our diagramatic convention. First of all, we let time flow up. Each edge in a diagram is an anyon propagator. An upward arrow on a propagator signifies a nontrivial anyon, whereas flipping the arrow turns the anyon to its anti-anyon. The trivial sector (vacuum) takes either a dashed or dotted line (to be clarified shortly), without an arrow. Any ``horizontal" propagator connecting two vertical propagators must be tilted in such a way that its left end is later in time than its right end. The following definition of $F$ symbols illustrates our convention.
\[
\Hgraph[1]{a}{b}{c}{d}{e}=\sum_f \Fm{a}{b}{c}{d}{e}{f}{}{}\Xgraph[1]{a}{b}{c}{d}{f}.
\]  
Moreover, we shall exclusively deal with self-dual anyon models in this paper, namely the models where an anyon is the anti-anyon of itself. As such, the arrows are herein redundant; however, for the sake of future extensions, we keep the arrows, without explicitly denoting the anti-anyons.

In our breaking scheme of quantum group symmetry, topological sectors of $\A$ are rearranged, or branch into those of $\T\supset\U$. 
A propagator of an anyon $t$ in $\T$ is thus mapped to linear combinations of  $\A$ propagators involving all the lifts of the anyon $t$, namely
\be\label{eq:propagatorLift}
\propa[1]{t}{1}=\sum_{a\in t}\propa[1]{a}{1},
\ee 
where an oriented vertical line is a propagator, and the notion $a\in t$ refers to all the lifts of $t$. A propagator is also thought of as an object in the corresponding fusion category. We denote the $\T$ vacuum by $\varphi$ and draw its propagator as a dashed line without an arrow, and we draw a dotted line for the propagator of the true vacuum in $\A$; hence, the following is understood. 
\be\label{eq:propagatorVacLift}
\propa{-2}{}\; =\;\propa[1]{\varphi}{} \;=\; \propa{-1}{}\; +\; \propa[1]{\gamma}{0},
\ee 
where $\gamma$ is the only simple current condensate in $\A$, which is the case we consider exclusively in this paper. Therefore, one expects that there exists a linear relation between the Hilbert space of $\T$ and that of $\A$. In particular, because such a Hilbert space is spanned by the states corresponding to the 3-point vertices, one asserts that
\be\label{eq:vlc}
\Yup[1]{r}{s}{t}{\mu}=\sum_{a\in r, b\in s, c\in t}\vlc{r}{s}{t}{a}{b}{c}_{\mu} \Yup[1]{a}{b}{c}{},
\ee
where the complex coefficients $\vlc{r}{s}{t}{a}{b}{c}_{\mu}$ are called the vertex lifting coefficients (VLCs), which should depend on the topological sectors  in $\T$ connected at a vertex and their lifts in $\A$. Here we remark that although we assume to restrict our discussion to $\A$ whose fusion algebra is multiplicity free, i.e., $N^c_{ab}=\delta_{abc}$, the fusion algebra of the broken phase $\T$ may have multiplicity, as to be seen in Section \ref{sec:su2_8break} where $\A=SU(2)_8$. Thus, on the LHS of Eq. \eqref{eq:vlc}, a generic vertex of $\T$ carries a multiplicity index $\mu=1,\dots, N^t_{rs}$, while on the RHS of the equation, the VLCs carry the same multiplicity index. This emergent multiplicity in $\T$'s fusion algebra is not treated in Ref.\cite{Eliens2013}, and we generalize the procedures there to accommodate these situations. Note that vertices with different multiplicity labels represent orthogonal basis states, namely
\[
\Blangle \Ydown[1]{r}{s}{t}{\mu} \Bvert \Yup[1]{r}{s}{t'}{\nu} \Brangle= \bubbleGraph{t}{t'}{r}{s}{\mu}{\nu}= \sqrt \frac{d_t}{d_rd_s}\delta_{tt'}  \delta_{\mu\nu}\propa[1]{t}{1}\; .
\]
Hence, if we think of $\vlc{r}{s}{t}{a}{b}{c}_{\mu}$ for given $r$, $s$, and $t$ as a vector whose elements are labeled by $abc$ as a single index, then the following orthogonality is obvious.
\be
\sum_{abc\in rst} \vlc{r}{s}{t}{a}{b}{c}_{\mu}\vlc{r}{s}{t}{a}{b}{c}_{\nu}^*=0, \quad \text{for }\mu\neq \nu.
\ee 
For simplicity, we adopt the convention in Ref.\cite{Eliens2013} to give the following special VLCs more convenient notations.
\be\label{eq:tdef}
t^{\gamma a}_b:=\sqrt{\kappa}\vlc{\varphi}{t}{t}{\gamma}{a}{b},\;     t^{a \gamma}_b := \sqrt{\kappa} \vlc{t}{\varphi}{t}{a}{\gamma}{b}, \; t^{ab}_\gamma:=\sqrt{\kappa} \vlc{t}{t}{\varphi}{a}{b}{\gamma}
\ee
where $\kappa=d_1+d_\gamma$ is a normalization factor. In general, $\kappa$ is the sum of the quantum dimensions of all sectors in $\A$ that are identified with the vacuum $\varphi$ in $\T$. Note that unitarity forbids multiplicity to appear in Eq. \eqref{eq:tdef}. We choose the following gauge of the braiding between a nontrivial $\T$ sector $t$ and the $\T$ vacuum\cite{Eliens2013}:
\be
\VRxoss[0]{1}{t}{-2}{t}{}\equiv \Yup[1]{t}{-2}{t}{}\Longrightarrow t^{a\gamma}_{b} \equiv R^{\gamma a}_b t^{\gamma a}_b.
\ee
Besides, the VLCs for $\T$ vertices with all three legs being the $\T$ vacuum are even more special and can be conveniently denoted as follows.
\be
\phi^i_{jk}=\sqrt\kappa \vlc{\varphi}{\varphi}{\varphi}{\gamma_i}{\gamma_j}{\gamma_k},\quad \gamma_i,\ \gamma_j,\ \gamma_k\in\{1,\gamma\},
\ee
where $\gamma\in\A$ is the only simple current condensate we consider. It turns out that there are consistency conditions that demands $\phi^i_{jk}=1$ $\forall i,j,k$ allowed by $\A$'s fusion algebra. The general VLCs can be obtained via imposing the consistency conditions \eqref{eq:ABCconsistency} of the elementary vertices, whose algebraic meaning is exhibited in Eq. \eqref{eq:ABCeqs}.
\be\label{eq:ABCconsistency}
\Yupst{\mu}= \Yup[1]{r}{s}{t}{\mu},\quad \Yuprt{\mu}= \Yup[1]{r}{s}{t}{\mu},\quad \Yupsr{\mu}= \Yup[1]{r}{s}{t}{\mu}.
\ee
\begin{subequations}\label{eq:ABCeqs}
\begin{alignat}{2}
&\sum\limits_{b',c'}A\indices*{*^{b'}_b^{c'}_c}{\vlc{r}{s}{t}{a}{b'}{c'}}_\mu &&= {\vlc{r}{s}{t}{a}{b}{c}}_\mu ,\label{eq:Aeq}\\
&\sum\limits_{a',c'}B\indices*{*^{a'}_a^{c'}_c}{\vlc{r}{s}{t}{a'}{b}{c'}}_\mu &&= {\vlc{r}{s}{t}{a}{b}{c}}_\mu,\label{eq:Beq}\\
&\sum\limits_{a',b'}C\indices*{*^{a'}_a^{b'}_b}{\vlc{r}{s}{t}{a'}{b'}{c}}_\mu &&= {\vlc{r}{s}{t}{a}{b}{c}}_\mu,\label{eq:Ceq}
\end{alignat}
\end{subequations}
where
\begin{subequations}\label{eq:ABCmatrices}
\begin{alignat}{2}
& A\indices*{*^{b'}_b^{c'}_c} &&=\frac{1}{\kappa}\sum_{\gamma\in\varphi}(s^{b'\gamma }_b)^* t^{c'\gamma}_c{\Fm{c'}{\gamma}{a}{b}{b'}{c}}^*\sqrt{\frac{d_{c'}d_\gamma}{d_c}}
,\label{eq:Amatrix} \\
& B\indices*{*^{a'}_a^{c'}_c} &&=\frac{1}{\kappa}\sum_{\gamma\in\varphi}(r^{\gamma a'}_a)^* t^{\gamma c'}_{c}\Fm{a}{b}{\gamma}{c'}{a'}{c}\sqrt{ \frac{d_{c'}d_\gamma}{d_c}}, \label{eq:Bmatrix}\\
& C\indices*{*^{a'}_a^{b'}_b} &&=\frac{1}{\kappa}\sum_{\gamma\in\varphi}(r^{a'\gamma}_a)^* s^{\gamma b}_{b'}\Fm{a}{b}{a'}{b'}{\gamma}{c}\sqrt{\frac{d_{a'}d_{b'}}{d_c}}.\label{eq:Cmatrix}
\end{alignat}
\end{subequations}
The above matrices $A$, $B$, and $C$ are independent of the multiplicity. One then sees that a VLC vector $\vlc{r}{s}{t}{a}{b}{c}$ for given $r$, $s$, and $t$ is the common $+1$ eigenvector of the three matrices in Eq. \eqref{eq:ABCmatrices}. The general solution for such a $+1$ eigenvector may contain a few unknown parameters. There are normalization conditions (Eqs. B4 to B6 in Ref.\cite{Eliens2013}) that can partially fix these parameters, such that only certain phase parameters remain. One can then use two conditions (Eqs. B13 and B14 in Ref.\cite{Eliens2013}) to further fix the relative phases between VLC vectors that are related by permutation of the sectors connected at a vertex. In the end, one ends up with an overall phase parameter for each VLC vector, which cannot be fixed because of the freedom in $\T$ to redefine elementary vertices by a phase\cite{Eliens2013}. Note that this overall phase freedom should not be confused with the gauge freedom that one can rescale each leg of a vertex by a phase, i.e. $\vlc{t}{v}{r}{a}{b}{c}\rightarrow \vlc{t}{v}{r}{a}{b}{c} \alpha^r_a\alpha^s_b/\alpha^t_c$. Nevertheless, since we consider self-dual anyon models only, namely the models with $a=\bar{a}$, these phases $\alpha$ must be merely $\pm 1$. Clearly, we can gauge-fix this freedom by invoking only the special VLCs defined in Eq. \eqref{eq:tdef}, leading to a constraint for self-dual models:
\[
\frac{t^{\gamma a}_{b}}{{t^{\gamma a}_b}^*}=R^{\gamma b}_a{\Fm{1}{a}{\gamma}{b}{\gamma}{a}}^* \Fm{1}{\gamma}{b}{a}{b}{\gamma}\Fm{b}{a}{\gamma}{1}{a}{\gamma}.
\] 

The $F$ symbols of $\T$ are expressible in terms of those of $\A$ and the VLC's, as follows\footnote{This equation corrects the original equation in Ref. \cite{Eliens2013}, which mistakenly sums over the index $f$ rather than $d$.}.
\be\label{eq:FmatrixT}
\begin{aligned}
\Fm{r}{s}{t}{u}{v,\mu\nu}{w,\mu'\nu'}=&\sum_{abcde}\vlc{t}{v}{r}{c}{e}{a}_\mu^* \vlc{v}{s}{u}{e }{b}{d}_\nu \vlc{r}{s}{w}{a}{b}{f}_{\mu'}^*\vlc{t}{u}{w}{c}{d}{f}_{\nu'}  \\
&\x \sqrt{\frac{d_ad_bd_cd_d}{d_rd_sd_td_u}}\frac{d_w}{d_f}\Fm{a}{b}{c}{d}{e}{f},\quad \forall f\in w.
\end{aligned}
\ee
 From Eq. \eqref{eq:FmatrixT}, the consistency equations \eqref{eq:ABCeqs}, Eqs. \eqref{eq:ABCmatrices}, and the normalization conditions (((to be added in appendix)))   one can find the following identities of $\T$'s $F$ symbols.
\begin{align}
\Fm{r}{\varphi}{t}{u}{u,\mu 1}{r,1\nu}=\Fm{r}{s}{\varphi}{u}{r,1\nu}{u,\mu 1}= \delta_{\mu \nu} \\
\Fm{r}{s}{r}{s}{\varphi,11}{w,\mu\nu}=\sqrt\frac{d_w}{d_rd_s}\delta_{\mu\nu},
\end{align}
where the script $1$ labels the vertices that do not have multiplicity greater than one. One can actually infer these identities directly from the consistency conditions in Eq. \eqref{eq:ABCconsistency}.

\subsection{Braiding in the $\T$ phase}
Although the $\T$ phase might not have a Braided Tensor Category (BTC) description, in certain cases, we can still define the braiding operators such as the $R$ and monodromy $M$ operators that act on the topological sectors of the $\T$ phase, in particular the cases where at least one confined sector is involved. In fact, because two unconfined sectors, namely sectors in the bare $\U$ phase can never fuse to a confined sector, a $\T$ vertex under our consideration involves either only one $\U$ sector or three $\U$ sectors. Since the bare $\U$ phase is a BTC and the corresponding topological data has been generally studied in Ref. \cite{Eliens2013}, we shall not dwell on it in this paper. 

We would like to define monodromies between anyons in $\T$ via $\A$. By means of the VLCs, we first lift a $\T$ vertex to a linear combination of $\A$ vertices, then the monodromy acting on the $\T$ vertex is defined to act on each of the $\A$ vertex in this combination. The monodromy therefore rotates the $\T$ state embedded in $\A$ to some other state in $\A$ (this may not be a state of $\T$ any longer).
i.e.
\be\label{eq:DefTmonodromy}
\begin{aligned}
M_\T\BLvert\Yup[1]{r}{s}{t}{\mu}\Brangle &=M_\A\sum_{a,b,c}\vlc{r}{s}{t}{a}{b}{c}_\mu \BLvert\Yup[1]{a}{b}{c}{}\Brangle\\
& =\sum_{a,b,c}\vlc{r}{s}{t}{a}{b}{c}_\mu M^{ab}_c \BLvert\Yup[1]{a}{b}{c}{}\Brangle\\
& =\sum_{a,b,c}\vlc{r}{s}{t}{a}{b}{c}_\mu' \BLvert\Yup[1]{a}{b}{c}{}\Brangle,
\end{aligned}
\ee  
where 
\be
\vlc{r}{s}{t}{a}{b}{c}_\mu'=\vlc{r}{s}{t}{a}{b}{c}_\mu M^{ab}_c =
\vlc{r}{s}{t}{a}{b}{c}_\mu\frac{\theta_c}{\theta_a\theta_b}.
\ee
Note that here as the VLCs $\vlc{r}{s}{t}{a}{b}{c}_\mu$ for fixed $r$, $s$, and $t$ are arranged in a vector indexed by the triple (abc), the monodromy operator $M_\A$ is always diagonalized in this basis. Besides, since vertices with different multiplicity indices are orthogonal in the Hilbert space, multiplicity mixing will not arise under monodromy actions.

If the new linear combination in the end of Eq. \eqref{eq:DefTmonodromy} turns out to be proportional to the lift of some $\T$ vertex, we can "{\bf pull back}" the combination to the $\T$ vertex potentially also with a phase factor. This then defines the matrix element $M^{rs}_t$ of $M_\T$. Nevertheless, a pullback does not always exist because the RHS of Eq. \eqref{eq:DefTmonodromy} may not be interpreted as  a state of  $\T$, or in other words, the lift of any $\T$ vertex. A natural but crucial question is under what condition the action of monodromy on the lift of a $\T$ vertex can be pulled back. Unless $M_\A\propto \mathds{1}$, whose action can always be pulled back, otherwise a general answer is not obvious. We now try to partially answer this question in several steps, assuming $M_\A\not\propto\mathds{1}$.

The monodromy operator $M_\A$ acts on a vector space whose dimension is determined by the number of lifts of the corresponding $\T$ vertex. Let us denote this space by $V^{rst}_\A$. This space is decomposable if any of the $\T$ sectors has more than one lifts. Suppose $r$ has multiple lifts $a$, we have the decomposition $V^{rst}_\A=\oplus_a V^{rst}_a$, where the lifts $b$, $c$ of $s$, $t$ are neglected as they label the basis in each subspace $V^{rst}_a$. We refer to this decomposition of $V^{rst}_\A$ the decomposition with respect to $r$'s lifts, and each subspace $V^{rst}_a$ is labeled by a particular lift $a$ of $r$. There may also be decompositions with respect to $s$'s and/or $t$'s lifts. One may wonder whether such a subspace is invariant up to an overall phase under the action of $M_\A$. The following proposition answers this question.
\begin{proposition}\label{lem:invSubspace}
For a space $V^{rst}_\A$, take any one of three legs to decompose the space, if the other two legs are sectors neither both in $\U$ nor both in $\T\setminus\U$, the subspaces in this decomposition are uninvariant subspaces of the monodromy operator $M_\A$ on $V^{rst}_\A$.   
\end{proposition}
\begin{proof}
Without loss of generality, let us take $r$ with lifts $\{a_1,a_2,\dots,a_i,\dots\}$ to decompose the space, and suppose $s\in\U$ whereas $t\in\T\setminus\U$. Since clearly all subspaces in the decomposition are either invariant or uninvariant simultaneously, we can consider a generic subspace $V^{rst}_{a_i}$ of $V^{rst}_\A$, since $M_\A$ is diagonal on $V^{rst}_\A$, the block of $M_\A$ that acts on $V^{rst}_{a_i}$ reads $\diag\{\theta_{c_1}/(\theta_{a_i}\theta_{b_1}),\dots, \theta_{c_j}/(\theta_{a_i}\theta_{b_j}),\dots\}$. Since $s$ is unconfined, it has a well-defined topological spin, and thus any two of its lifts, say $b_j$ and $b_k$ must satisfy $\theta_{b_j}=\theta_{b_k}$. On the contrary, because $t$ is confined, according to Ref.\cite{Eliens2013}, for any two lifts of $t$, say $c_j$ and $c_k$, $h_{c_j}-h_{c_k}\not\in\Z$; hence, $\theta_{c_j} \neq\theta_{c_k}$. Thus, $\theta_{c_j}/(\theta_{a_i}\theta_{b_j})\neq \theta_{c_k} /(\theta_{a_i}\theta_{b_k})$. That is, $V^{rst}_{a_i}$ is not invariant under $M_\A$.   
\end{proof}

One can infer from the proof above that if a subspace in a decomposition of $V^{rst}_\A$ is an invariant subspace, all other subspaces in the same decomposition are also invariant; however, these subspaces are invariant up to different phases. Since the rotation $M_\A$ does not affect the basis vertices in the lift of a $\T$ but only the VLCs, and the lift of a $\T$ vertex is  unique up to merely an overall phase, now that if the $V^{rst}_\A$ being rotated has subspaces invariant up to distinct phases, the action of $M_\A$ cannot turn the lift of the $\T$ vertex into the same lift or the lift of another $\T$ vertex up to an overall phase. This argument leads to the following necessary condition of pullback.
\begin{theorem}\label{thm:TMonodromyNC}
If the action of $M_\A$ on the lift of a $\T$ vertex has a pullback to certain $\T$ vertex (not necessarily a different one), then the lifted space $V^{rst}_\A$ of the $\T$ vertex admits no decomposition into invariant subspaces with respect to any of the three legs $r$, $s$, and $t$.
\end{theorem}
We have not settled on a most general sufficient condition of pullback, which however, may not exist because accidental symmetries in the anyon spectra may often exist in various topological phases, in particular the nonchiral ones. At this point, we are only able to claim that the action of $M_\A$ on the lift of any $\U$ vertex can be pulled back. A simple reasoning is that as long as we begin with a quantum group equivalent to a modular tensor category, we would end up with another modular tensor category $\U$ via our quantum group symmetry breaking scheme, and that modular tensor categories have well-behaving modular $S$ and $T$ matrices. In fact, one can directly write down the modular $S$ matrix of the $\U$ phase\footnote{The idea of deriving this $S$ matrix partially follows Eq. (86) in Ref.\cite{Eliens2013}; however, we choose to write the $S$ matrix of $\U$ explicitly in terms of the relevant VLCs and $\T$'s $F$ symbols, as this way is more transparent in the derivation.} in terms of the VLCs and part of $\A$'s topological data:
\be\label{eq:SmatrixU}
S_{st}=\frac{1}{\kappa D_\U}\sum_r \Fm{s}{t}{s}{t}{\varphi}{r}\sum_{abc\in rst}\vlc{s}{t}{r}{a}{b}{c}\vlc{s}{t}{r}{a}{b}{c}^* \sqrt{d_ad_bd_c}\frac{\theta_c}{\theta_a\theta_b},
\ee
where $D_\U=\sqrt{\sum_{u\in\U} d^2_u}$ is the total quantum dimension of the $\U$ phase. Note that the above equation contains no multiplicity indices, as the multiplicity does not arise for $\U$ vertices for our consideration in this paper. Similarly, we can also try to define the $R$ matrices for the $\T$ phase by means of pullback. That is, we first do the following.
\be\label{eq:defTRmatrix}
\begin{aligned}
\BLvert\VRxoss[0]{1}{r}{s}{t}{\mu}\Brangle &=\sum_{a,b,c}\vlc{s}{r}{t}{b}{a}{c}_\mu \BLvert\VRxoss[0]{1}{a}{b}{c}{}\Brangle\\
& =\sum_{a,b,c}\vlc{s}{r}{t}{b}{a}{c}_\mu R^{ba}_c \BLvert\Yup[1]{a}{b}{c}{}\Brangle\\
& \stackrel{?}{\propto}\sum_{a,b,c}\vlc{r'}{s'}{t'}{a}{b}{c}_\mu\BLvert\Yup[1]{r'}{s'}{t'}{}\Brangle,
\end{aligned}
\ee
where $R^{ba}_c$ is the $R_\A$ matrix elements of $\A$, and $r's't'$ are not necessarily different from $rst$. If the proportionality in the last equation indeed holds, we can define the proportionality factor to be the corresponding matrix element of the $R_\T$. As in the case with defining $M_\T$, there may not exist a sufficient condition that the $R_\A$ on the lift of a $\T$ vertex can be pulled back; however, a necessary condition is similar to Theorem \ref{thm:TMonodromyNC}. As a remark, that a $R_\T$ matrix element exists for a $\T$ vertex does not guarantee that the corresponding $M_\T$ matrix element exist for the vertex. Roughly speaking, $R$ matrices are not gauge invariant operators as opposed to monodromy matrices; hence, it is reasonable that there are fewer pullbackable actions of $M_\A$ on the lifts of $\T$ vertices than the pullbackable actions of $R_\A$ on the corresponding $\T$ vertices. More precisely speaking, that a $R_\A$ action on the lift of a $\T$ vertex can be pulled back implies that there are no invariant subspaces in the decomposition of the lifted vector space with respect to any leg of the vertex; however, since $M_\A$ is essentially $R_\A^2$, it is likely that the uninvariant subspaces under the action of the $R_\A$ turns out to be invariant under the action of the $M_\A$, thus not pullbackable any more. We also remark that since $M_\T$ generates symmetry transformation on the unconfined particles for the group element $g$ associated to the confined particle being braided around, we conclude that $M_T^n$ necessarily has a pullback for $g^n = 1$. In the following two sections, we shall see examples supporting these facts.

\section{$\Z_2$ symmetric $\Z_2$ toric code from \II breaking}\label{sec:Ising2break}
This example has been discussed in \cite{Hung2013}, and a connection between the theories has been foreshadowed  in \cite{Bais2002,Bais2003} and also proposed in \cite{Lu2013}. In our previous discussion \cite{Hung2013}, we have made a guess of the mapping between the Hilbert space of the broken theory and the parent theory, with which we demonstrate how we can define the braiding and monodromy between an unconfined particle and a confined one. This mapping of the Hilbert space is precisely defined by VLCs in the framework of \cite{Eliens2013}, which was almost identical to our guess apart from a treatment of the confined particles. In any event, in the following, we would like to revisit the example adopting the formalism of \cite{Eliens2013}, which has the virtue of making associativity of the fusion algebra in the broken theory explicit. The topological data of the \II theory is reviewed in the appendix. The anyon content of the \II theory is as listed in table \ref{tab:isingany}. To recover the $\Z_2$ symmetric toric code model, the appropriate ``condensate''  is taken as $(\psi,\psi)$. The anyons of the broken theory and their relationship with the parent \II theory is listed in table \ref{tab:isingTany}. 

\begin{table}[h!]
\centering
{
\setlength{\extrarowheight}{1.5pt}
\begin{tabular}{m{5em}|m{6em}| m{7em}|}\hline
&Anyon  in $\A $             & Corresponding anyons in $\T$         \\ \hline
unconfined&$(1,1) ~ (\psi,\psi) $             & $1$      \\ \hline 
&$(\sigma,\sigma)$   & $e \oplus m$       \\ \hline
&$(\psi,1)~ (1,\psi)$          & $f$        \\ \hline 
\hline
confined&$(1,\sigma) ~ (\psi,\sigma)$    & $\chi$\\ \hline
&$(\sigma,1)~(\sigma,\psi)$         &  $\tilde\chi$ \\ \hline
\end{tabular}}
\caption{\II anyon content and their relationship with the broken theory.} \label{tab:isingTany}
\end{table}  
The unconfined sectors correspond precisely to those of the toric code model. As we discussed in detail in the previous section, the symmetry transformation properties of the anyons can be read off from their monodromy property with a chosen confined anyon. Without loss of generality, we can pick $\chi$ to be the symmetry generating anyon. All the VLCs mapping each $\T$ 3-point vertex back to a combination of $\A$ 3-point vertices, are solved and listed in the appendix. These data allow us to recover the symmetry transformation properties of our anyons, using Eq. \eqref{eq:DefTmonodromy}.  Let us note here an important difference in our considerations from that of Ref.\cite{Eliens2013}. As we have repeatedly emphasized, in the context of an SEP, the ``condensate'', while belonging to the trivial topological sector, carries non-trivial global symmetry charge. Therefore, unlike the description of true anyon condensation where the condensate is strictly indistinguishable from the trivial sector, here, we have to keep track of the difference between the multiple lifts in $\A$ of an anyon in $\T$. This is how the table is obtained above, where, for instance, the phase accumulated in the monodromy of $(1,\psi)$ around a confined particle is clearly different from that of the $(\psi,1)$. This however underlines a missing ingredient in the current framework in dealing with SEP's: we need a systematic way to treat the distinction between the different lifts of a single $\T$ anyon.  As already brought up in the previous section, the different charges of the lifts  in $\A$ of an anyon in $\T$ leads to VLC's rotated in orthogonal directions after a monodromy with confined particles, which has no interpretation as states in the Hilbert space of $\T$ in the current framework. Therefore this is very suggestive that to treat the Hilbert space of the SEP as opposed to $\T$, we should include these new sets of VLC's generated by monodromy and take them as a map of valid states in the SEPs to states in $\A$.  We are currently working on a more complete and rigorous Mathematical treatment of the problem via the tools of graded tensor category. These will be reported elsewhere.

The symmetry transformation properties of the unconfined sectors are given in table \ref{tab:isingUtrans}.
\begin{table}[h!]
\centering
{
\setlength{\extrarowheight}{1.5pt}
\begin{tabular}{|m{6em}| m{10em}|}\hline
 Anyons  & Symmetry transformation generated by $\chi$ (overall phase unless specified)      \\ \hline
$(1,1)\in \mathbf{1}  $             &  untransformed     \\ \hline 
$(\psi,\psi)\in\mathbf{1}  $             &  $-1$    \\ \hline 
$(\psi,1)\in\epsilon$          &  $-1$    \\ \hline 
$(1,\psi)\in\epsilon$          &   $1$    \\ \hline 
$(\sigma,\sigma)= e \oplus m$   & $-\e^{\ii \pi/4}  \sigma_x $       \\ \hline
\end{tabular}}
\caption{Symmetry transformation properties of the unconfined sector. $\sigma_x$ is the Pauli matrix denoting the transformation that exchanges the two sectors $e$ and $m$.} \label{tab:isingUtrans}
\end{table}  

Let us also demonstrate the last row in Table \ref{tab:isingUtrans}, which is the most significant property of the $\Z_2$ toric code with a nonlocally realized $\Z_2$ global symmetry, by an explicit application of Eq. \eqref{eq:DefTmonodromy}.
\begin{align}
M_\T\BLvert\Yup[1]{u_\pm}{\chi}{\tilde\chi}{}\Brangle  =&M_\A\sum_{b\in\chi,c\in\tilde\chi} \vlc{u_\pm}{\chi}{\tilde\chi}{\ss}{b}{c} \BLvert\Yup[1]{\ss}{b}{c}{}\Brangle\nonumber\\
=&\pm \frac{\e^{-\frac{\ii\pi}{4}}\theta_{\si}}{\theta_{\ss}\theta_{\is}} \BLvert\Yup[1]{\ss}{\is}{\si}{}\Brangle + \frac{\theta_{\spsi}}{\theta_{\ss}\theta_{\is}} \BLvert \Yup[1]{\ss}{\is}{\spsi}{}\Brangle\nonumber\\
&+\frac{\theta_{\si}}{\theta_{\ss}\theta_{\psis}}\BLvert\Yup[1]{\ss}{\psis}{\si}{}\Brangle \mp\frac{\e^{\frac{\ii\pi}{4}}\theta_{\spsi}}{\theta_{\ss}\theta_{\psis}} \BLvert\Yup[1]{\ss}{\psis}{\spsi}{}\Brangle\nonumber\\
=&-\e^{\frac{\ii\pi}{4}}\left\{\mp\e^{-\frac{\ii\pi}{4}}\BLvert\Yup[1]{\ss}{\is}{\si}{}\Brangle +\BLvert \Yup[1]{\ss}{\is}{\spsi}{}\Brangle\right.\nonumber\\
&\left. +\BLvert\Yup[1]{\ss}{\psis}{\si}{}\Brangle\pm\e^{\frac{\ii\pi}{4}} \BLvert\Yup[1]{\ss}{\psis}{\spsi}{}\Brangle\right\}\nonumber\\
=&-\e^{\frac{\ii\pi}{4}}\BLvert\Yup[1]{u_\mp}{\chi}{\tilde\chi}{}\Brangle,
\end{align}
where $u_+=e$ and $u_-=m$ are understood as in Appendix \ref{app:Ising2vlc}, and the last equality manifests the pullback. Such an exchange symmetry between charge and flux in this type of $\Z_2$-symmetric $\Z_2$ spin liquid has been also observed and studied via introducing topological defects\cite{Bombin2010,Kitaev2012,Barkeshli2012a,You2012a,You2013,Teo2013,Teo2013a,Barkeshli,Barkeshli2013} into the $\Z_2$ spin liquid, crossing which a chargeon (fluxon) turns into a fluxon (chargeon). 

In addition to symmetry transformation properties, it is also of interest to inspect properties of the confined sector. We note that following \cite{Eliens2013}, one can extract the $F$ matrix in the $\T$ theory describing $F$-moves of four defects. We note that $1,f, \chi$ (or similarly ($1,f,\tilde{\chi}$)) forms a fusion sub-algebra and the $F$ matrices $F^{\chi\chi}_{\chi\chi}$ and $F^{\tilde\chi\tilde\chi}_{\tilde\chi\tilde\chi}$ agree precisely with $F^{\sigma\sigma}_{\sigma\sigma}$ in the Ising model:
\be
F^{\chi\chi}_{\chi\chi} =F^{\tilde\chi\tilde\chi}_{\tilde\chi\tilde\chi} = \frac{1}{\sqrt{2}} \begin{pmatrix} 
1&1\\
1&-1\end{pmatrix}.
\ee

\section{$\Z_2$ symmetric chiral Fibonacci$\x$Fibonacci Model from $SU(2)_8$ breaking}\label{sec:su2_8break}
In this section, we construct and study the  Fibonacci $\times$ Fibonacci phase enriched by $\Z_2$ symmetry by breaking a non-Abelian topological phase with the hidden quantum group symmetry $U_q(SU(2))$, where $q=\exp(\ii\pi/5)$. According to the relation between Wess-Zumino-Witten (WZW) theory and quantum groups, the unitary, irreducible representations of a quantum group $U_q(G)$, i.e., the topological sectors of the corresponding topological phase, are in one-to-one correspondence with the chiral primary fields of the WZW theory at level $k$, with central charge $c=k\mathrm{dim}_G/(k+\hat g)$, where $\hat g$ is the dual Coxter label of $G$. The relation between $q$ and $k$ reads $q=\exp(\ii 2\pi/(k+\hat g))$. The chiral primary fields of  the WZW theory at level $k$ are in fact the unitary, irreducible representations of the chiral algebra $G_k$. As such, for   $U_q(SU(2))$ with $q=\exp(\ii\pi/5)$, the corresponding chiral algebra is $SU(2)_8$. The fusion and braiding of the topological sectors of the quantum group are precisely those of the chiral primary fields of the corresponding WZW theory. Therefore, we can study the breaking of $U_q(SU(2))$ equivalently via the breaking of $SU(2)_8$. 

The topological data of $SU(2)_k$ is reviewed in Appendix \ref{app:su2_kData}. Note that the vacuum of $\A=SU(2)_8$ is conventionally denoted by $0$ instead of 1. To break this quantum group symmetry, we condense the only simple current sector, namely the sector $8$ in the spectrum. According to the fusion algebra and topological spins of all the $SU(2)_8$ sectors, one could obtain the broken phase $\T$ that contains both confined and unconfined sectors, and the phase $\U$ that has the unconfined sectors only, which are summarized in Table \ref{tab:su2_8TUany} below \cite{Bais2009a}.  
\begin{table}[h!]
\centering
{
\setlength{\extrarowheight}{2pt}
\begin{tabular}{m{4em}| m{6em}| m{5em}| m{2em}| m{5em}}\hline
$\A$ sectors               & $\T$ sectors         & $d$ & $h$ &  \\ \hline
$0,\;8\rightarrow$             & $\mathbf{1}$   & 1 &  0   & unconfined\\ \hline 
$2,\;6\rightarrow$        & $(\tau,\tau)$ & $\phi^2$ & $\frac{1}{5}$ & unconfined \\ \hline
$\ \ \; \; 4\rightarrow$          & $(1,\tau)+(\tau,1)$  & $\phi$ & $\frac{3}{5}$   & unconfined\\ \hline
$1,\;7\rightarrow$        & $\zeta$  & $\sqrt{\frac{5+\sqrt{5}}{2}}$ & & confined\\ \hline
$3,\;5\rightarrow$   & $\tilde\zeta$         & $\sqrt{5+2\sqrt{5}}$ & & confined\\ \hline
\end{tabular}}
\caption{$SU(2)_8$ breaking: anyon content. Where $\phi=\frac{1+\sqrt{5}}{2}$. Note that only unconfined sectors have well defined topological spins.}
\label{tab:su2_8TUany}
\end{table}

The topological sectors of the $\U$ phase form the fusion algebra of the chiral Fibonacci $\x$ Fibonacci model, namely, $(1,\tau)\x(1,\tau)=\mathbf{1}+(1,\tau)$, $(\tau,1)\x(\tau,1)=\mathbf{1}+(\tau,1)$, $(1,\tau)\x(\tau,1)=(\tau,\tau)$, and $(\tau,\tau)^2=\mathbf{1}+(1,\tau)+(\tau,1)+(\tau,\tau)$. In contrast to the previous example, the unbroken bare phase $\U$ is a non-Abelian phase. Intuitively, we can actually almost infer that the symmetry on this bare phase $\U$ is yet again $\Z_2$ because as in the previous example, here the lifts of $\T$'s vacuum---the true vacuum sector $0$ and the charged condensate $8$---form a $\Z_2$ fusion subalgebra of $\A's$ fusion algbera, and the fusion of any two confined sectors in Table \ref{tab:su2_8TUany} falls in $\U$, implying also a $\Z_2$ structure. We can make this intuition precise by compute the VLCs and hence the topological properties of the $\T$ and $\U$ phases, in particular extracting the global symmetry action on the $\U$ phase from the braiding between the unconfined sectors and any representative of the confined ones. Appendix \ref{app:su2_8vlc} collects for this case all the VLCs, and the well-defined $R$ and $M$ matrix elements of the $\T$ and $\U$ phases.

As we have repeatedly emphasized and explicitly shown in the previous section, the actual SEP obtained (i.e. a bare phase $\U$ enriched by a global symmetry) from breaking the phase $\A$ by a charged condensate should not be confused with the $\T$ phase. SEP has a larger Hilbert space than $\T$ which requires extra sets of VLC to account for the distinction between a given topological sector with different global charges, which are by definition identified in $\T$. In the SEP so obtained, any basis state should have well-defined transformation properties under the global symmetry, in the sense that the monodromy $M_\A$ on its lift can be pulled back.

Using the prescription for monodromies discussed in the previous sections, and taking $\zeta$ as the representative of the confined sectors that generates the global symmetry, we  obtain the following table that presents the charges of the various sectors in the $\Z_2$ symmetric Fibonacci $\x$ Fibonacci phase embedded in $SU(2)_8$.
\begin{table}[h!]
\centering
{
\setlength{\extrarowheight}{1.5pt}
\begin{tabular}{| m{8em}| m{7em}|}\hline
 Anyons  & Symmetry transformation generated by $\zeta$ (overall phase unless specified)      \\ \hline
$0\in \mathbf{1}  $             &  untransformed     \\ \hline 
$8\in \mathbf{1} $             &  $-1$    \\ \hline 
$2\in (\tau,\tau)$          &  $1$    \\ \hline 
$6\in (\tau,\tau)$          &   $-1$    \\ \hline 
$4= (1,\tau) \oplus (\tau,1)$   & $\e^{-\ii 3 \pi/5} \sigma_x$   \\ \hline
\end{tabular}}
\caption{Symmetry transformation properties of the unconfined sector in $SU(2)_8$ breaking. $\sigma_x$ is the Pauli matrix denoting the transformation that exchanges the two sectors $(1,\tau)$ and $(\tau,1)$.} 
\label{tab:FiboUtrans}
\end{table} 

In Table \ref{tab:FiboUtrans}, we keep track of  the different lifts of each Fibonacci $\x$ Fibonacci sector and find the $\Z_2$ charges they carry. Like in the case with \II breaking, here the differently charged lifts of $(\tau,\tau)$ lead to different VLCs and hence different linear combinations of $SU(2)_8$ vertices that are transformed orthogonally under the monodromy with a confined sector, such as $\zeta$, and thus result in new linear combinations of $SU(2)_8$ vertices that have no interpretation in the Hilbert space of the $\T$ phase. Again, this suggests we may take the coefficients in the new linear combinations as also valid VLCs in the SET which has a larger Hilbert space than $\T$.  The complete Hilbert space of the SET, cannot be directly obtained in the current approach of computing the VLCs of $\T$. At this moment, however, we are still working on a more Mathematically rigorous and systematic framework based on graded tensor categories.

\section{Discussions and Outlook}\label{sec:disc}
In this paper, we continue along the path started in \cite{Hung2013} in embedding
SEPs in a larger topological theory. We made use of the framework and tools
developed in \cite{Eliens2013} that describes anyon condensation and adopt it
for our purpose of extracting symmetry and topological data of the SEP from
a parent topological phase. We demonstrate with a few explicit examples
how the non-Abelian SETs can be reconstructed from a parent theory. We note however,
that the current framework is still incomplete, because it strictly identifies anyons in the
parent theory related by fusion with the  ``condensate'' which 
carries non-trivial global symmetry charge. A complete description of SEP should involve enlarging
the Hilbert space of the condensed phase. A rigorous Mathematical treatment will involve
graded tensor category theory, which we will report elsewhere in a forthcoming publication.

\begin{acknowledgements}
We thank Davide Gaiotto, Yuting Hu, Tian Lan, Juven Wang. LYH is supported by the Croucher Fellowship. This research was supported in part by Perimeter Institute for
Theoretical Physics. Research at Perimeter Institute is
supported by the Government of Canada through Industry
Canada and by the Province of Ontario through the
Ministry of Economic Development \& Innovation.
\end{acknowledgements}

\begin{appendix}
\section{\II topological data}\label{app:Ising2ModData}
In this appendix we construct the \II topological data from that of the Ising topological order. 
\subsection{Ising topological data}
We hereby first collect the topological data of the Ising topological order. Ising topological order contains three topological sectors $\{1,\sigma,\psi\}$, with quantum dimeansion $\{1,\sqrt{2},1\}$, topological spins $\{h_a=0,h_\sigma=\tfrac{1}{16},h_\psi=\tfrac{1}{2}\}$, and fusion group $\{\sigma^2=1+\psi,\psi^2=1,\sigma\psi=\psi\sigma=\sigma\}$. Note that each Ising anyon is the anti-anyon of itself. The nontrivial $R$ matrix entries are $R^{\ss}_1=\e^{-\ii\tfrac{\pi}{8}}$, $R^{\pp}_1=-1$, $R^{\psis}_\sigma=R^{\spsi}_\sigma=-\ii$, and $R^{\ss}_\psi=\e^{\ii\tfrac{3\pi}{8}}$. The only nonzero $F$ matrix entries are as follows.
\be\label{eq:IsingF}
\begin{aligned}
&\Fm{\sigma}{\sigma}{\sigma}{\sigma}{1}{1}=
\Fm{\sigma}{\sigma}{\sigma}{\sigma}{1}{\psi}=
\Fm{\sigma}{\sigma}{\sigma}{\sigma}{\psi}{1}=
-\Fm{\sigma}{\sigma}{\sigma}{\sigma}{\psi}{\psi}=\frac{1}{\sqrt{2}}\\
&\Fm{\sigma}{\psi}{\sigma}{\psi}{1}{\sigma}=
\Fm{\psi}{\sigma}{\psi}{\sigma}{1}{\sigma}=
\Fm{1}{\sigma}{1}{\sigma}{1}{\sigma}=\Fm{\sigma}{1}{\sigma}{1}{1}{\sigma}=1\\
&\Fm{\sigma}{\sigma}{\psi}{\psi}{\sigma}{1}=\Fm{\psi}{\psi}{\sigma}{\sigma}{\sigma}{1}
=\Fm{\sigma}{\sigma}{1}{1}{\sigma}{1}=\Fm{1}{1}{\sigma}{\sigma}{\sigma}{1}=1\\
&\Fm{\sigma}{\psi}{\psi}{\sigma}{\sigma}{\sigma}=\Fm{\psi}{\sigma}{\sigma}{\psi}{\sigma}{\sigma}=-1\\
&\Fm{\sigma}{1}{1}{\sigma}{\sigma}{\sigma}=\Fm{1}{\sigma}{\sigma}{1}{\sigma}{\sigma}=
\Fm{\psi}{1}{1}{\psi}{\psi}{\psi}=\Fm{1}{\psi}{\psi}{1}{\psi}{\psi}=1\\
&\Fm{\sigma}{1}{\sigma}{\psi}{\psi}{\sigma}=\Fm{\sigma}{\psi}{\sigma}{1}{\psi}{\sigma}
=\Fm{1}{\sigma}{\psi}{\sigma}{\psi}{\sigma}=\Fm{\psi}{\sigma}{1}{\sigma}{\psi}{\sigma}=1\\
&\Fm{1}{\psi}{\sigma}{\sigma}{\sigma}{\psi}=\Fm{\sigma}{\sigma}{1}{\psi}{\sigma}{\psi}
=\Fm{\psi}{1}{\sigma}{\sigma}{\sigma}{\psi}=\Fm{\sigma}{\sigma}{\psi}{1}{\sigma}{\psi}=1\\
&\Fm{\sigma}{1}{\psi}{\sigma}{\sigma}{\sigma}=\Fm{\psi}{\sigma}{\sigma}{1}{\sigma}{\sigma}
=\Fm{1}{\sigma}{\sigma}{\psi}{\sigma}{\sigma}=\Fm{\sigma}{\psi}{1}{\sigma}{\sigma}{\sigma}=1\\
&\Fm{1}{1}{\psi}{\psi}{\psi}{1}=\Fm{\psi}{\psi}{1}{1}{\psi}{1}=\Fm{1}{\psi}{1}{\psi}{1}{\psi}
=\Fm{\psi}{1}{\psi}{1}{1}{\psi}=1\\
&\Fm{1}{1}{1}{1}{1}{1}=\Fm{\psi}{\psi}{\psi}{\psi}{1}{1}=1.
\end{aligned}
\ee

\subsection{Topological data of \II}
The dual Ising topological order, $\overline{\text{Ising}}$ is the same as the Ising one except that all the topological spins are flipped. This, however, does not change the $S$ matrix and $F$ matrices but only complex conjugates the $R$ matrix. Usually the sectors in $\overline{\text{Ising}}$ are denoted by $\{1,\bar\sigma,\bar\psi\}$; however, for our convenience, we would denote the usual Ising sectors by $(a,1)$ and the $\overline{\text{Ising}}$ sectors by $(1, a)$, with $a\in\{1,\sigma,\psi\}$, but keeping in mind that $(1,a)$ has topological spins $\{0,-\tfrac{1}{16},-\tfrac{1}{2}\}$. As such, $(a,b)$ is a topological sector of \II, with quantum dimension $d_{(a,b)}=d_ad_b$ and topological spin $h_{(a,b)}=h_a+h_b$. The fusion algebra is just the direct product of two copies of the Ising fusion algebra, namely $(a,b)\x(a',b')=(a\x a',b\x b')$. There are nine sectors in total that are tabulated as follows.
\begin{table}[h!]
\centering
{
\setlength{\extrarowheight}{1.5pt}
\begin{tabular}{m{3em}| m{2em}| m{3em}}\hline
Anyon               & $d$         & $h$  \\ \hline
$(1,1)$             & $1$         & $0$\\ \hline 
$(1,\sigma)$        & $\sqrt{2}$  & $-1/16$ \\ \hline
$(1,\psi)$          & $1$         & $-1/2$\\ \hline
$(\sigma,1)$        & $\sqrt{2}$  & $1/16$\\ \hline
$(\sigma,\sigma)$   & $2$         & $0$\\ \hline
$(\sigma,\psi)$     & $\sqrt{2}$  & $-7/16$\\ \hline
$(\psi,1)$          & $1$         & $1/2$\\ \hline
$(\psi,\sigma)$     & $\sqrt{2}$  & $7/16$\\ \hline
$(\psi,\psi)$       & $1$         & $0$\\ \hline
\end{tabular}}\label{tab:isingany}
\caption{\II anyon content.}
\end{table}
We also denote an \II anyon by $(ab)$ or simply $ab$ wherever no ambiguity arises. The $R$ matrix elements of two \II sectors $ab$ and $a'b'$ fusing to $cd$ reads $R^{(ab)(a'b')}_{(cd)}=R^{aa'}_c \overline{R^{bb'}_d}$, where the $\bar{R}$ is the complex conjugate of the $R$ value of the Ising. An $F$ matrix of the \II is the tensor product of the two $F$ matrices of the corresponding anyons of the Ising phase, and entry-wise this means
\be
\Fm{(aa')}{(bb')}{(cc')}{(dd')}{(ee')}{(ff')}=\Fm{a}{b}{c}{d}{e}{f} \Fm{a'}{b'}{c'}{d'}{e'}{f'},
\ee
which can be easily obtained from Eq. \eqref{eq:IsingF}. 

\section{VLC's for \II breaking}\label{app:Ising2vlc}
This appendix records the VLCs for the \II breaking by condensing the sector $(\psi,\bar\psi)$. Please refer to Table \ref{tab:isingTany} for the notations that appear in below. Note that $u_+=e$ and $u_-=m$ in the $\Z_2$ toric code are understood.
\be
(u_\pm)^{(\pp)(\ss)}_{(\ss)}=(u_\pm)^{(\ss)(\pp)}_{(\ss)}=(u_\pm)^{(\ss)(\ss)}_{(\pp)}=\pm 1.
\ee
\be
\epsilon^{(\pp)(\psii)}_{(\ip)}=\epsilon^{(\ip)(\pp)}_{(\psii)}= \epsilon^{(\ip)(\psii)*}_{(\pp)}=\ii.
\ee
\be
\epsilon^{(\psii)(\pp)}_{(\ip)}=\epsilon^{(\pp)(\ip)}_{(\psii)}=\epsilon^{(\psii)(\ip)*}_{(\pp)}=-\ii.
\ee
\be
\begin{aligned}
\vlc{u_\pm}{\chi}{\tchi}{\ss}{\is}{\si} &=\pm \e^{-\ii\pi/4},\\
\vlc{u_\pm}{\chi}{\tchi}{\ss}{\is}{\spsi} &=1,\\ 
\vlc{u_\pm}{\chi}{\tchi}{\ss}{\psis}{\si} &=1,\\
\vlc{u_\pm}{\chi}{\tchi}{\ss}{\psis}{\spsi} &= \mp\e^{\ii\pi/4},
\end{aligned}
\ee
\be
\begin{aligned}
\vlc{\chi}{u_\pm}{\tchi}{\is}{\ss}{\si} &=\pm \e^{\ii\pi/4},\\
\vlc{\chi}{u_\pm}{\tchi}{\is}{\ss}{\spsi} &=1,\\ 
\vlc{\chi}{u_\pm}{\tchi}{\psis}{\ss}{\si} &=1,\\
\vlc{\chi}{u_\pm12}{\tchi}{\psis}{\ss}{\spsi} &= \mp\e^{-\ii\pi/4},
\end{aligned}
\ee
\be
\begin{aligned}
\vlc{u_\pm}{\tchi}{\chi}{\ss}{\si}{\is} &=\pm \e^{\ii\pi/4},\\
\vlc{u_\pm}{\tchi}{\chi}{\ss}{\si}{\psis} &=1,\\ 
\vlc{u_\pm}{\tchi}{\chi}{\ss}{\spsi}{\is} &=1,\\
\vlc{u_\pm}{\tchi}{\chi}{\ss}{\spsi}{\psis} &= \mp\e^{-\ii\pi/4},
\end{aligned}
\ee
\be
\begin{aligned}
\vlc{\tchi}{u_\pm}{\chi}{\si}{\ss}{\is} &=\pm \e^{-\ii\pi/4},\\
\vlc{\tchi}{u_\pm}{\chi}{\si}{\ss}{\psis} &=1,\\ 
\vlc{\tchi}{u_\pm}{\chi}{\spsi}{\ss}{\is} &=1,\\
\vlc{\tchi}{u_\pm}{\chi}{\spsi}{\ss}{\psis} &= \mp\e^{\ii\pi/4},
\end{aligned}
\ee
\be
\begin{aligned}
\vlc{\epsilon}{\chi}{\chi}{\ip}{\is}{\is} &=1,\\
\vlc{\epsilon}{\chi}{\chi}{\ip}{\psis}{\psis} &=-1,\\ 
\vlc{\epsilon}{\chi}{\chi}{\psii}{\is}{\psis} &=e^{\ii\pi/4},\\
\vlc{\epsilon}{\chi}{\chi}{\psii}{\psis}{\is} &= e^{-\ii\pi/4},
\end{aligned}
\ee
\be
\begin{aligned}
\vlc{\chi}{\epsilon}{\chi}{\is}{\ip}{\is} &= 1,\\
\vlc{\chi}{\epsilon}{\chi}{\is}{\psii}{\psis} &=\e^{-\ii\pi/4},\\ 
\vlc{\chi}{\epsilon}{\chi}{\psis}{\ip}{\psis} &=-1,\\
\vlc{\chi}{\epsilon}{\chi}{\psis}{\psii}{\is} &= \e^{\ii\pi/4},
\end{aligned}
\ee
\be
\begin{aligned}
\vlc{\epsilon}{\tchi}{\tchi}{\ip}{\si}{\spsi} &= \e^{-\ii\pi/4},\\
\vlc{\epsilon}{\tchi}{\tchi}{\ip}{\spsi}{\si} &=\e^{\ii\pi/4},\\ 
\vlc{\epsilon}{\tchi}{\tchi}{\psii}{\si}{\si} &=1,\\
\vlc{\epsilon}{\tchi}{\tchi}{\psii}{\spsi}{\spsi} &= -1,
\end{aligned}
\ee
\be
\begin{aligned}
\vlc{\tchi}{\epsilon}{\tchi}{\si}{\ip}{\spsi} &= e^{\ii\pi/4},\\
\vlc{\tchi}{\epsilon}{\tchi}{\si}{\psii}{\si} &=1,\\ 
\vlc{\tchi}{\epsilon}{\tchi}{\spsi}{\is}{\si} &=e^{-\ii\pi/4},\\
\vlc{\tchi}{\epsilon}{\tchi}{\spsi}{\psii}{\spsi} &= -1,
\end{aligned}
\ee
\be
\chi^{(\pp)(\is)}_{(\psis)}=\e^{-\ii\pi/4},\quad \chi^{(\is)(\pp)}_{(\psis)}=\e^{\ii\pi/4}.
\ee
\be
\chi^{(\pp)(\psis)}_{(\is)}=\e^{\ii\pi/4},\quad \chi^{(\psis)(\pp)}_{(\is)}=\e^{-\ii\pi/4}.
\ee
\be
\tchi^{(\pp)(\si)}_{(\spsi)}=\e^{\ii\pi/4},\quad \tchi^{(\si)(\pp)}_{(\spsi)}=\e^{-\ii\pi/4}.
\ee
\be
\tchi^{(\pp)(\spsi)}_{(\si)}=\e^{-\ii\pi/4},\quad \tchi^{(\spsi)(\pp)}_{(\si)}=\e^{\ii\pi/4}.
\ee
\be
\chi^{(\is)(\psis)}_{(\pp)}=\e^{-\ii\pi/4},\quad \chi^{(\psis)(\is)}_{(\pp)}=\e^{\ii\pi/4}.
\ee
\be
\tchi^{(\si)(\spsi)}_{(\pp)}=\e^{\ii\pi/4},\quad \tchi^{(\spsi)(\si)}_{(\pp)}=\e^{-\ii\pi/4}.
\ee
\be
\begin{aligned}
\vlc{\chi}{\chi}{\epsilon}{\is}{\is}{\ip} &=1,\\
\vlc{\chi}{\chi}{\epsilon}{\is}{\psis}{\psii} &=\e^{\ii\pi/4},\\ 
\vlc{\chi}{\chi}{\epsilon}{\psis}{\is}{\psii} &=\e^{-\ii\pi/4},\\
\vlc{\chi}{\chi}{\epsilon}{\psis}{\psis}{\ip} &=-1,
\end{aligned}
\ee
\be
\begin{aligned}
\vlc{\tchi}{\tchi}{\epsilon}{\si}{\si}{\psii} &= 1,\\
\vlc{\tchi}{\tchi}{\epsilon}{\si}{\spsi}{\ip} &=\e^{-\ii\pi/4},\\ 
\vlc{\tchi}{\tchi}{\epsilon}{\spsi}{\si}{\ip} &=\e^{\ii\pi/4},\\
\vlc{\tchi}{\tchi}{\epsilon}{\spsi}{\spsi}{\psii} &=-1,
\end{aligned}
\ee
\be
\begin{aligned}
\vlc{\chi}{\tchi}{u_\pm}{\is}{\si}{\ss} &= \pm\e^{-\ii\pi/4},\\
\vlc{\chi}{\tchi}{u_\pm}{\is}{\spsi}{\ss} &=1,\\ 
\vlc{\chi}{\tchi}{u_\pm}{\psis}{\si}{\ss} &=1,\\
\vlc{\chi}{\tchi}{u_\pm}{\psis}{\spsi}{\ss} &=\mp\e^{\ii\pi/4},
\end{aligned}
\ee
\be
\begin{aligned}
\vlc{\tchi}{\chi}{u_\pm}{\si}{\is}{\ss} &= \pm\e^{\ii\pi/4},\\
\vlc{\tchi}{\chi}{u_\pm}{\si}{\psis}{\ss} &=1,\\ 
\vlc{\tchi}{\chi}{u_\pm}{\spsi}{\is}{\ss} &=1,\\
\vlc{\tchi}{\chi}{u_\pm}{\spsi}{\psis}{\ss} &=\mp\e^{-\ii\pi/4},
\end{aligned}
\ee

\section{$U_q(SU(2))$ Topological data}\label{app:su2_kData}
In this appendix we record the formulae to obtain the topological data of the quantum groups $U_q(SU(2))$ with $q$ roots of unity. For each $q$ value, the quantum group has irreducible representations in one-to-one correspondence with those of the chiral algebra $SU(2)_k$ with $q=\e^{\frac{\ii 2\pi}{k+2}}$. Although we only dealt with the case with $k=8$ in the main part of the paper, we include the formulae for general $k\in\Z$ for completeness and future reference. As these formulae are well-known mathematical results, one may find them in the literature, e.g., Ref.\cite{Eliens2013}.  

$SU(2)_k$ has $k$ distinct topological sectors, namely $0,1,\dots,k$, where $0$ is the trivial or vacuum sector. The fusion algebra of these sectors reads
\be\label{eq:su2kfusion}
a\x b=c_{ab}+(c_{ab}+2)+\cdots+\min\{a+b,2k-a-b\},
\ee
where $c_{ab}=|a-b|$. Clearly this is a truncated tensor product of the usual $SU(2)$ irreducible representations. The multiplicity $N^c_{ab}=1$ if $|a-b|\leq c\leq \min \{a+b,2k-a-b\}$, $a+b+c=0\pmod{2}$, and $a+b+c\leq 2k$; otherwise, $N^c_{ab}=0$. A generic sector $a$ has quantum dimension and self-statistics respectively
\be\label{eq:su2kqdTheta}
\begin{aligned}
d_a=\frac{\sin(\frac{a+1}{k+2}\pi)}{\sin\frac{\pi}{k+2}},\\
\theta_a=\e^{\ii 2\pi\frac{a(a+2)}{4(k+2)}}.
\end{aligned}
\ee
The $R$ matrix elements are
\be\label{eq:su2kR}
R^{ab}_c=\ii^{c-a-b}q^{\frac{1}{8}[c(c+2)-a(a+2)-b(b+2)]}.
\ee
The $F$ symbols are given by the general formula:
\be\label{eq:su2kF}
\Fm{a}{b}{c}{d}{e}{f}=\ii^{a+b+c+d}\sqrt{\frac{d_ed_f}{d_ad_d}[a+1]_q[d+1]_q} \sixj{c}{e}{a}{b}{f}{d}^*,
\ee
where the $6j$ symbols
\begin{align*}
\sixj{c}{e}{a}{b}{f}{d}&=\Delta(c,e,a)\Delta(e,c,d)\Delta(e,b,d)\Delta(c,d,f)\\
\x \sum_z\Biggl\{ & \frac{(-1)^z[z+1]_q!}{[z-\frac{c+e+a}{2}]_q![z-\frac{a+b+f}{2}]_q! [z-\frac{e+b+d}{2}]_q!} \\
  &\x\frac{1}{[z-\frac{c+d+f}{2}]_q![\frac{c+e+b+f}{2}-z]_q!}\\
  &\x\frac{1}{[\frac{c+a+b+d}{2}-z]_q![\frac{e+a+f+d}{2}-z]_q!}\Biggr\}
\end{align*}
are defined with
\[
\Delta(a,b,c)=\sqrt\frac{[\frac{-a+b+c}{2}]_q![\frac{a-b+c}{2}]_q![\frac{a+b-c}{2}]_q!} {[\frac{a+b+c}{2}+1]_q!},
\]
which is invariant under permutation of its variables, and the $q$-numbers and $q$-factorials
\[
[n]_q=\frac{q^{n/2}-q^{-n/2}}{q^{1/2}-q{-1/2}},\quad [n]_q!=\prod_{m=1}^n[m]_q.
\]
By definition $[0]_q!\equiv 1$. Note that the sum over $z$ in the above expression of the $6j$ symbols is carried from $\max\{\frac{c+e+a}{2},\frac{a+b+f}{2},\frac{e+b+d}{2}, \frac{c+d+f}{2}\}$ to $\min\{\frac{c+e+b+f}{2},\frac{c+a+b+d}{2},\frac{e+a+f+d}{2}\}.$
\section{$SU(2)_8$ breaking VLCs}\label{app:su2_8vlc}
In what follows we catalog all the VLCs for $SU(2)_8$ breaking by condensing the sector $8$.

First we note that $t^{0 a}_b=t^{a0}_b=t^{ab}_0=\delta_{ab},\ \forall a,b\in t$. The other nonvanishing VLCs are listed as follows.
\be
\begin{aligned}
&\zeta^{81}_7=\zeta^{78}_1=\zeta^{17}_8=\frac{1+\ii}{\sqrt 2},\\ 
&\zeta^{87}_1 =\zeta^{18}_7=\zeta^{71}_8=-\frac{1-\ii}{\sqrt 2}.
\end{aligned}
\ee
\be
\begin{aligned}
&(\tau\tau)^{82}_6=(\tau\tau)^{68}_2=(\tau\tau)^{26}_8=\ii,\\
&(\tau\tau)^{86}_2=(\tau\tau)^{28}_6=(\tau\tau)^{62}_8=-\ii.
\end{aligned}
\ee
\be
\begin{aligned}
&\tilde\zeta^{83}_5=\tilde\zeta^{58}_3=\tilde\zeta^{35}_8=\frac{1-\ii}{\sqrt 2},\\
&\tilde\zeta^{85}_3=\tilde\zeta^{38}_5=\tilde\zeta^{53}_8=-\frac{1+\ii}{\sqrt 2}.
\end{aligned}
\ee
\be
\begin{aligned}
&(1\tau)^{84}_4=(1\tau)^{48}_4=(1\tau)^{44}_8=1,\\
&(\tau 1)^{84}_4=(\tau 1)^{48}_4=(\tau 1)^{44}_8=-1.
\end{aligned}
\ee
\be
\begin{aligned}\vlc{\zeta}{\zeta}{\tau \tau}{1}{1}{2}&=-\frac{1}{\sqrt{2}},\\
\vlc{\zeta}{\zeta}{\tau \tau}{1}{7}{6}&=-\frac{(-1)^{3/4}}{\sqrt{2}},\\
\vlc{\zeta}{\zeta}{\tau \tau}{7}{1}{6}&=-\frac{\sqrt[4]{-1}}{\sqrt{2}},\\
\vlc{\zeta}{\zeta}{\tau \tau}{7}{7}{2}&=\frac{1}{\sqrt{2}},
\end{aligned}
\ee
\be
\begin{aligned}\vlc{\zeta}{\tau \tau}{\zeta}{1}{2}{1}&=-\frac{\sqrt[4]{-1}}{\sqrt{2}},\\
\vlc{\zeta}{\tau \tau}{\zeta}{1}{6}{7}&=-\frac{i}{\sqrt{2}},\\
\vlc{\zeta}{\tau \tau}{\zeta}{7}{2}{7}&=\frac{\sqrt[4]{-1}}{\sqrt{2}},\\
\vlc{\zeta}{\tau \tau}{\zeta}{7}{6}{1}&=\frac{1}{\sqrt{2}},
\end{aligned}
\ee
\be
\begin{aligned}\vlc{\zeta}{\tau \tau}{\tilde{\zeta }}{1}{2}{3}&=\frac{\sqrt[4]{-1}}{\sqrt{2}},\\
\vlc{\zeta}{\tau \tau}{\tilde{\zeta }}{1}{6}{5}&=\frac{1}{\sqrt{2}},\\
\vlc{\zeta}{\tau \tau}{\tilde{\zeta }}{7}{2}{5}&=-\frac{(-1)^{3/4}}{\sqrt{2}},\\
\vlc{\zeta}{\tau \tau}{\tilde{\zeta }}{7}{6}{3}&=\frac{1}{\sqrt{2}},
\end{aligned}
\ee
\be
\begin{aligned}\vlc{\zeta}{\tilde{\zeta }}{\tau \tau}{1}{3}{2}&=-\frac{i}{\sqrt{2}},\\
\vlc{\zeta}{\tilde{\zeta }}{\tau \tau}{1}{5}{6}&=\frac{(-1)^{3/4}}{\sqrt{2}},\\
\vlc{\zeta}{\tilde{\zeta }}{\tau \tau}{7}{3}{6}&=\frac{(-1)^{3/4}}{\sqrt{2}},\\
\vlc{\zeta}{\tilde{\zeta }}{\tau \tau}{7}{5}{2}&=\frac{1}{\sqrt{2}},
\end{aligned}
\ee
\be
\begin{aligned}\vlc{\zeta}{\tilde{\zeta }}{1\tau}{1}{3}{4}&=\frac{i}{2^{3/4}},\\
\vlc{\zeta}{\tilde{\zeta }}{1\tau}{1}{5}{4}&=-\frac{\sqrt[4]{-1}}{2^{3/4}},\\
\vlc{\zeta}{\tilde{\zeta }}{1\tau}{7}{3}{4}&=\frac{\sqrt[4]{-1}}{2^{3/4}},\\
\vlc{\zeta}{\tilde{\zeta }}{1\tau}{7}{5}{4}&=\frac{1}{2^{3/4}},
\end{aligned}
\ee
\be
\begin{aligned}\vlc{\zeta}{\tilde{\zeta }}{\tau 1}{1}{3}{4}&=\frac{i}{2^{3/4}},\\
\vlc{\zeta}{\tilde{\zeta }}{\tau 1}{1}{5}{4}&=\frac{\sqrt[4]{-1}}{2^{3/4}},\\
\vlc{\zeta}{\tilde{\zeta }}{\tau 1}{7}{3}{4}&=-\frac{\sqrt[4]{-1}}{2^{3/4}},\\
\vlc{\zeta}{\tilde{\zeta }}{\tau 1}{7}{5}{4}&=\frac{1}{2^{3/4}},
\end{aligned}
\ee
\be
\begin{aligned}\vlc{\zeta}{1\tau}{\tilde{\zeta }}{1}{4}{3}&=-\frac{i}{2^{3/4}},\\
\vlc{\zeta}{1\tau}{\tilde{\zeta }}{1}{4}{5}&=-\left(-\frac{1}{2}\right)^{3/4},\\
\vlc{\zeta}{1\tau}{\tilde{\zeta }}{7}{4}{3}&=\left(-\frac{1}{2}\right)^{3/4},\\
\vlc{\zeta}{1\tau}{\tilde{\zeta }}{7}{4}{5}&=\frac{1}{2^{3/4}},
\end{aligned}
\ee
\be
\begin{aligned}\vlc{\zeta}{\tau 1}{\tilde{\zeta }}{1}{4}{3}&=-\frac{i}{2^{3/4}},\\
\vlc{\zeta}{\tau 1}{\tilde{\zeta }}{1}{4}{5}&=\left(-\frac{1}{2}\right)^{3/4},\\
\vlc{\zeta}{\tau 1}{\tilde{\zeta }}{7}{4}{3}&=-\left(-\frac{1}{2}\right)^{3/4},\\
\vlc{\zeta}{\tau 1}{\tilde{\zeta }}{7}{4}{5}&=\frac{1}{2^{3/4}},
\end{aligned}
\ee
\be
\begin{aligned}\vlc{\tau \tau}{\zeta}{\zeta}{2}{1}{1}&=\frac{(-1)^{3/4}}{\sqrt{2}},\\
\vlc{\tau \tau}{\zeta}{\zeta}{2}{7}{7}&=-\frac{(-1)^{3/4}}{\sqrt{2}},\\
\vlc{\tau \tau}{\zeta}{\zeta}{6}{1}{7}&=-\frac{i}{\sqrt{2}},\\
\vlc{\tau \tau}{\zeta}{\zeta}{6}{7}{1}&=-\frac{1}{\sqrt{2}},
\end{aligned}
\ee
\be
\begin{aligned}\vlc{\tau \tau}{\zeta}{\tilde{\zeta }}{2}{1}{3}&=\frac{(-1)^{3/4}}{\sqrt{2}},\\
\vlc{\tau \tau}{\zeta}{\tilde{\zeta }}{2}{7}{5}&=-\frac{\sqrt[4]{-1}}{\sqrt{2}},\\
\vlc{\tau \tau}{\zeta}{\tilde{\zeta }}{6}{1}{5}&=\frac{1}{\sqrt{2}},\\
\vlc{\tau \tau}{\zeta}{\tilde{\zeta }}{6}{7}{3}&=\frac{1}{\sqrt{2}},
\end{aligned}
\ee
\be
\begin{aligned}\vlc{\tau \tau}{\tau \tau}{\tau \tau}{2}{2}{2}&=-\frac{1}{\sqrt{2}},\\
\vlc{\tau \tau}{\tau \tau}{\tau \tau}{2}{6}{6}&=\frac{1}{\sqrt{2}},\\
\vlc{\tau \tau}{\tau \tau}{\tau \tau}{6}{2}{6}&=\frac{1}{\sqrt{2}},\\
\vlc{\tau \tau}{\tau \tau}{\tau \tau}{6}{6}{2}&=\frac{1}{\sqrt{2}},
\end{aligned}
\ee
\be
\begin{aligned}\vlc{\tau \tau}{\tau \tau}{1\tau}{2}{2}{4}&=\frac{i}{2^{3/4}},\\
\vlc{\tau \tau}{\tau \tau}{1\tau}{2}{6}{4}&=-\frac{1}{2^{3/4}},\\
\vlc{\tau \tau}{\tau \tau}{1\tau}{6}{2}{4}&=\frac{1}{2^{3/4}},\\
\vlc{\tau \tau}{\tau \tau}{1\tau}{6}{6}{4}&=\frac{i}{2^{3/4}},
\end{aligned}
\ee
\be
\begin{aligned}\vlc{\tau \tau}{\tau \tau}{\tau 1}{2}{2}{4}&=\frac{1}{2^{3/4}},\\
\vlc{\tau \tau}{\tau \tau}{\tau 1}{2}{6}{4}&=-\frac{i}{2^{3/4}},\\
\vlc{\tau \tau}{\tau \tau}{\tau 1}{6}{2}{4}&=\frac{i}{2^{3/4}},\\
\vlc{\tau \tau}{\tau \tau}{\tau 1}{6}{6}{4}&=\frac{1}{2^{3/4}},
\end{aligned}
\ee
\be
\begin{aligned}\vlc{\tau \tau}{\tilde{\zeta }}{\zeta}{2}{3}{1}&=\frac{\sqrt[4]{-1}}{\sqrt{2}},\\
\vlc{\tau \tau}{\tilde{\zeta }}{\zeta}{2}{5}{7}&=-\frac{(-1)^{3/4}}{\sqrt{2}},\\
\vlc{\tau \tau}{\tilde{\zeta }}{\zeta}{6}{3}{7}&=\frac{1}{\sqrt{2}},\\
\vlc{\tau \tau}{\tilde{\zeta }}{\zeta}{6}{5}{1}&=\frac{1}{\sqrt{2}},
\end{aligned}
\ee
\be
\begin{aligned}\vlc{\tau \tau}{\tilde{\zeta }}{\tilde{\zeta }}{2}{3}{3}_1&=-\frac{\sqrt[4]{-1}}{\sqrt{2}},\\
\vlc{\tau \tau}{\tilde{\zeta }}{\tilde{\zeta }}{2}{3}{5}_1&=0,\\
\vlc{\tau \tau}{\tilde{\zeta }}{\tilde{\zeta }}{2}{5}{3}_1&=0,\\
\vlc{\tau \tau}{\tilde{\zeta }}{\tilde{\zeta }}{2}{5}{5}_1&=\frac{\sqrt[4]{-1}}{\sqrt{2}},\\
\vlc{\tau \tau}{\tilde{\zeta }}{\tilde{\zeta }}{6}{3}{3}_1&=0,\\
\vlc{\tau \tau}{\tilde{\zeta }}{\tilde{\zeta }}{6}{3}{5}_1&=-\frac{i}{\sqrt{2}},\\
\vlc{\tau \tau}{\tilde{\zeta }}{\tilde{\zeta }}{6}{5}{3}_1&=\frac{1}{\sqrt{2}},\\
\vlc{\tau \tau}{\tilde{\zeta }}{\tilde{\zeta }}{6}{5}{5}_1&=0,
\end{aligned}
\ee
\be
\begin{aligned}\vlc{\tau \tau}{\tilde{\zeta }}{\tilde{\zeta }}{2}{3}{3}_2&=0,\\
\vlc{\tau \tau}{\tilde{\zeta }}{\tilde{\zeta }}{2}{3}{5}_2&=\frac{\sqrt[4]{-1}}{\sqrt{2}},\\
\vlc{\tau \tau}{\tilde{\zeta }}{\tilde{\zeta }}{2}{5}{3}_2&=\frac{(-1)^{3/4}}{\sqrt{2}},\\
\vlc{\tau \tau}{\tilde{\zeta }}{\tilde{\zeta }}{2}{5}{5}_2&=0,\\
\vlc{\tau \tau}{\tilde{\zeta }}{\tilde{\zeta }}{6}{3}{3}_2&=-\frac{1}{\sqrt{2}},\\
\vlc{\tau \tau}{\tilde{\zeta }}{\tilde{\zeta }}{6}{3}{5}_2&=0,\\
\vlc{\tau \tau}{\tilde{\zeta }}{\tilde{\zeta }}{6}{5}{3}_2&=0,\\
\vlc{\tau \tau}{\tilde{\zeta }}{\tilde{\zeta }}{6}{5}{5}_2&=\frac{1}{\sqrt{2}},
\end{aligned}
\ee
\be
\begin{aligned}\vlc{\tau \tau}{1\tau}{\tau \tau}{2}{4}{2}&=\frac{1}{2^{3/4}},\\
\vlc{\tau \tau}{1\tau}{\tau \tau}{2}{4}{6}&=-\frac{i}{2^{3/4}},\\
\vlc{\tau \tau}{1\tau}{\tau \tau}{6}{4}{2}&=\frac{i}{2^{3/4}},\\
\vlc{\tau \tau}{1\tau}{\tau \tau}{6}{4}{6}&=\frac{1}{2^{3/4}},
\end{aligned}
\ee
\be
\begin{aligned}\vlc{\tau \tau}{1\tau}{\tau 1}{2}{4}{4}&=\frac{i}{\sqrt{2}},\\
\vlc{\tau \tau}{1\tau}{\tau 1}{6}{4}{4}&=\frac{1}{\sqrt{2}},
\end{aligned}
\ee
\be
\begin{aligned}\vlc{\tau \tau}{\tau 1}{\tau \tau}{2}{4}{2}&=\frac{1}{2^{3/4}},\\
\vlc{\tau \tau}{\tau 1}{\tau \tau}{2}{4}{6}&=\frac{i}{2^{3/4}},\\
\vlc{\tau \tau}{\tau 1}{\tau \tau}{6}{4}{2}&=-\frac{i}{2^{3/4}},\\
\vlc{\tau \tau}{\tau 1}{\tau \tau}{6}{4}{6}&=\frac{1}{2^{3/4}},
\end{aligned}
\ee
\be
\begin{aligned}\vlc{\tau \tau}{\tau 1}{1\tau}{2}{4}{4}&=-\frac{i}{\sqrt{2}},\\
\vlc{\tau \tau}{\tau 1}{1\tau}{6}{4}{4}&=\frac{1}{\sqrt{2}},
\end{aligned}
\ee
\be
\begin{aligned}\vlc{\tilde{\zeta }}{\zeta}{\tau \tau}{3}{1}{2}&=-\frac{i}{\sqrt{2}},\\
\vlc{\tilde{\zeta }}{\zeta}{\tau \tau}{3}{7}{6}&=-\frac{\sqrt[4]{-1}}{\sqrt{2}},\\
\vlc{\tilde{\zeta }}{\zeta}{\tau \tau}{5}{1}{6}&=-\frac{\sqrt[4]{-1}}{\sqrt{2}},\\
\vlc{\tilde{\zeta }}{\zeta}{\tau \tau}{5}{7}{2}&=-\frac{1}{\sqrt{2}},
\end{aligned}
\ee
\be
\begin{aligned}\vlc{\tilde{\zeta }}{\zeta}{1\tau}{3}{1}{4}&=-\frac{i}{2^{3/4}},\\
\vlc{\tilde{\zeta }}{\zeta}{1\tau}{3}{7}{4}&=\left(-\frac{1}{2}\right)^{3/4},\\
\vlc{\tilde{\zeta }}{\zeta}{1\tau}{5}{1}{4}&=-\left(-\frac{1}{2}\right)^{3/4},\\
\vlc{\tilde{\zeta }}{\zeta}{1\tau}{5}{7}{4}&=\frac{1}{2^{3/4}},
\end{aligned}
\ee
\be
\begin{aligned}\vlc{\tilde{\zeta }}{\zeta}{\tau 1}{3}{1}{4}&=-\frac{i}{2^{3/4}},\\
\vlc{\tilde{\zeta }}{\zeta}{\tau 1}{3}{7}{4}&=-\left(-\frac{1}{2}\right)^{3/4},\\
\vlc{\tilde{\zeta }}{\zeta}{\tau 1}{5}{1}{4}&=\left(-\frac{1}{2}\right)^{3/4},\\
\vlc{\tilde{\zeta }}{\zeta}{\tau 1}{5}{7}{4}&=\frac{1}{2^{3/4}},
\end{aligned}
\ee
\be
\begin{aligned}\vlc{\tilde{\zeta }}{\tau \tau}{\zeta}{3}{2}{1}&=\frac{(-1)^{3/4}}{\sqrt{2}},\\
\vlc{\tilde{\zeta }}{\tau \tau}{\zeta}{3}{6}{7}&=\frac{1}{\sqrt{2}},\\
\vlc{\tilde{\zeta }}{\tau \tau}{\zeta}{5}{2}{7}&=-\frac{\sqrt[4]{-1}}{\sqrt{2}},\\
\vlc{\tilde{\zeta }}{\tau \tau}{\zeta}{5}{6}{1}&=\frac{1}{\sqrt{2}},
\end{aligned}
\ee
\be
\begin{aligned}\vlc{\tilde{\zeta }}{\tau \tau}{\tilde{\zeta }}{3}{2}{3}_1&=\frac{(-1)^{3/4}}{\sqrt{2}},\\
\vlc{\tilde{\zeta }}{\tau \tau}{\tilde{\zeta }}{3}{2}{5}_1&=0,\\
\vlc{\tilde{\zeta }}{\tau \tau}{\tilde{\zeta }}{3}{6}{3}_1&=0,\\
\vlc{\tilde{\zeta }}{\tau \tau}{\tilde{\zeta }}{3}{6}{5}_1&=-\frac{i}{\sqrt{2}},\\
\vlc{\tilde{\zeta }}{\tau \tau}{\tilde{\zeta }}{5}{2}{3}_1&=0,\\
\vlc{\tilde{\zeta }}{\tau \tau}{\tilde{\zeta }}{5}{2}{5}_1&=-\frac{(-1)^{3/4}}{\sqrt{2}},\\
\vlc{\tilde{\zeta }}{\tau \tau}{\tilde{\zeta }}{5}{6}{3}_1&=-\frac{1}{\sqrt{2}},\\
\vlc{\tilde{\zeta }}{\tau \tau}{\tilde{\zeta }}{5}{6}{5}_1&=0,
\end{aligned}
\ee
\be
\begin{aligned}\vlc{\tilde{\zeta }}{\tau \tau}{\tilde{\zeta }}{3}{2}{3}_2&=0,\\
\vlc{\tilde{\zeta }}{\tau \tau}{\tilde{\zeta }}{3}{2}{5}_2&=\frac{(-1)^{3/4}}{\sqrt{2}},\\
\vlc{\tilde{\zeta }}{\tau \tau}{\tilde{\zeta }}{3}{6}{3}_2&=-\frac{1}{\sqrt{2}},\\
\vlc{\tilde{\zeta }}{\tau \tau}{\tilde{\zeta }}{3}{6}{5}_2&=0,\\
\vlc{\tilde{\zeta }}{\tau \tau}{\tilde{\zeta }}{5}{2}{3}_2&=\frac{\sqrt[4]{-1}}{\sqrt{2}},\\
\vlc{\tilde{\zeta }}{\tau \tau}{\tilde{\zeta }}{5}{2}{5}_2&=0,\\
\vlc{\tilde{\zeta }}{\tau \tau}{\tilde{\zeta }}{5}{6}{3}_2&=0,\\
\vlc{\tilde{\zeta }}{\tau \tau}{\tilde{\zeta }}{5}{6}{5}_2&=\frac{1}{\sqrt{2}},
\end{aligned}
\ee
\be
\begin{aligned}\vlc{\tilde{\zeta }}{\tilde{\zeta }}{\tau \tau}{3}{3}{2}_1&=-\frac{1}{\sqrt{2}},\\
\vlc{\tilde{\zeta }}{\tilde{\zeta }}{\tau \tau}{3}{3}{6}_1&=0,\\
\vlc{\tilde{\zeta }}{\tilde{\zeta }}{\tau \tau}{3}{5}{2}_1&=0,\\
\vlc{\tilde{\zeta }}{\tilde{\zeta }}{\tau \tau}{3}{5}{6}_1&=-\frac{\sqrt[4]{-1}}{\sqrt{2}},\\
\vlc{\tilde{\zeta }}{\tilde{\zeta }}{\tau \tau}{5}{3}{2}_1&=0,\\
\vlc{\tilde{\zeta }}{\tilde{\zeta }}{\tau \tau}{5}{3}{6}_1&=-\frac{(-1)^{3/4}}{\sqrt{2}},\\
\vlc{\tilde{\zeta }}{\tilde{\zeta }}{\tau \tau}{5}{5}{2}_1&=\frac{1}{\sqrt{2}},\\
\vlc{\tilde{\zeta }}{\tilde{\zeta }}{\tau \tau}{5}{5}{6}_1&=0,
\end{aligned}
\ee
\be
\begin{aligned}\vlc{\tilde{\zeta }}{\tilde{\zeta }}{\tau \tau}{3}{3}{2}_2&=0,\\
\vlc{\tilde{\zeta }}{\tilde{\zeta }}{\tau \tau}{3}{3}{6}_2&=-\frac{1}{\sqrt{2}},\\
\vlc{\tilde{\zeta }}{\tilde{\zeta }}{\tau \tau}{3}{5}{2}_2&=\frac{\sqrt[4]{-1}}{\sqrt{2}},\\
\vlc{\tilde{\zeta }}{\tilde{\zeta }}{\tau \tau}{3}{5}{6}_2&=0,\\
\vlc{\tilde{\zeta }}{\tilde{\zeta }}{\tau \tau}{5}{3}{2}_2&=\frac{(-1)^{3/4}}{\sqrt{2}},\\
\vlc{\tilde{\zeta }}{\tilde{\zeta }}{\tau \tau}{5}{3}{6}_2&=0,\\
\vlc{\tilde{\zeta }}{\tilde{\zeta }}{\tau \tau}{5}{5}{2}_2&=0,\\
\vlc{\tilde{\zeta }}{\tilde{\zeta }}{\tau \tau}{5}{5}{6}_2&=\frac{1}{\sqrt{2}},
\end{aligned}
\ee
\be
\begin{aligned}\vlc{\tilde{\zeta }}{\tilde{\zeta }}{1\tau}{3}{3}{4}&=\frac{i}{2^{3/4}},\\
\vlc{\tilde{\zeta }}{\tilde{\zeta }}{1\tau}{3}{5}{4}&=\frac{\sqrt[4]{-1}}{2^{3/4}},\\
\vlc{\tilde{\zeta }}{\tilde{\zeta }}{1\tau}{5}{3}{4}&=-\left(-\frac{1}{2}\right)^{3/4},\\
\vlc{\tilde{\zeta }}{\tilde{\zeta }}{1\tau}{5}{5}{4}&=\frac{i}{2^{3/4}},
\end{aligned}
\ee
\be
\begin{aligned}\vlc{\tilde{\zeta }}{\tilde{\zeta }}{\tau 1}{3}{3}{4}&=\frac{1}{2^{3/4}},\\
\vlc{\tilde{\zeta }}{\tilde{\zeta }}{\tau 1}{3}{5}{4}&=\left(-\frac{1}{2}\right)^{3/4},\\
\vlc{\tilde{\zeta }}{\tilde{\zeta }}{\tau 1}{5}{3}{4}&=\frac{\sqrt[4]{-1}}{2^{3/4}},\\
\vlc{\tilde{\zeta }}{\tilde{\zeta }}{\tau 1}{5}{5}{4}&=\frac{1}{2^{3/4}},
\end{aligned}
\ee
\be
\begin{aligned}\vlc{\tilde{\zeta }}{1\tau}{\zeta}{3}{4}{1}&=\frac{i}{2^{3/4}},\\
\vlc{\tilde{\zeta }}{1\tau}{\zeta}{3}{4}{7}&=\frac{\sqrt[4]{-1}}{2^{3/4}},\\
\vlc{\tilde{\zeta }}{1\tau}{\zeta}{5}{4}{1}&=-\frac{\sqrt[4]{-1}}{2^{3/4}},\\
\vlc{\tilde{\zeta }}{1\tau}{\zeta}{5}{4}{7}&=\frac{1}{2^{3/4}},
\end{aligned}
\ee
\be
\begin{aligned}\vlc{\tilde{\zeta }}{1\tau}{\tilde{\zeta }}{3}{4}{3}&=\frac{1}{2^{3/4}},\\
\vlc{\tilde{\zeta }}{1\tau}{\tilde{\zeta }}{3}{4}{5}&=-\frac{\sqrt[4]{-1}}{2^{3/4}},\\
\vlc{\tilde{\zeta }}{1\tau}{\tilde{\zeta }}{5}{4}{3}&=-\left(-\frac{1}{2}\right)^{3/4},\\
\vlc{\tilde{\zeta }}{1\tau}{\tilde{\zeta }}{5}{4}{5}&=\frac{1}{2^{3/4}},
\end{aligned}
\ee
\be
\begin{aligned}\vlc{\tilde{\zeta }}{\tau 1}{\zeta}{3}{4}{1}&=\frac{i}{2^{3/4}},\\
\vlc{\tilde{\zeta }}{\tau 1}{\zeta}{3}{4}{7}&=-\frac{\sqrt[4]{-1}}{2^{3/4}},\\
\vlc{\tilde{\zeta }}{\tau 1}{\zeta}{5}{4}{1}&=\frac{\sqrt[4]{-1}}{2^{3/4}},\\
\vlc{\tilde{\zeta }}{\tau 1}{\zeta}{5}{4}{7}&=\frac{1}{2^{3/4}},
\end{aligned}
\ee
\be
\begin{aligned}\vlc{\tilde{\zeta }}{\tau 1}{\tilde{\zeta }}{3}{4}{3}&=\frac{1}{2^{3/4}},\\
\vlc{\tilde{\zeta }}{\tau 1}{\tilde{\zeta }}{3}{4}{5}&=\frac{\sqrt[4]{-1}}{2^{3/4}},\\
\vlc{\tilde{\zeta }}{\tau 1}{\tilde{\zeta }}{5}{4}{3}&=\left(-\frac{1}{2}\right)^{3/4},\\
\vlc{\tilde{\zeta }}{\tau 1}{\tilde{\zeta }}{5}{4}{5}&=\frac{1}{2^{3/4}},
\end{aligned}
\ee
\be
\begin{aligned}\vlc{1\tau}{\zeta}{\tilde{\zeta }}{4}{1}{3}&=\frac{i}{2^{3/4}},\\
\vlc{1\tau}{\zeta}{\tilde{\zeta }}{4}{1}{5}&=-\frac{\sqrt[4]{-1}}{2^{3/4}},\\
\vlc{1\tau}{\zeta}{\tilde{\zeta }}{4}{7}{3}&=\frac{\sqrt[4]{-1}}{2^{3/4}},\\
\vlc{1\tau}{\zeta}{\tilde{\zeta }}{4}{7}{5}&=\frac{1}{2^{3/4}},
\end{aligned}
\ee
\be
\begin{aligned}\vlc{1\tau}{\tau \tau}{\tau \tau}{4}{2}{2}&=-\frac{1}{2^{3/4}},\\
\vlc{1\tau}{\tau \tau}{\tau \tau}{4}{2}{6}&=-\frac{i}{2^{3/4}},\\
\vlc{1\tau}{\tau \tau}{\tau \tau}{4}{6}{2}&=\frac{i}{2^{3/4}},\\
\vlc{1\tau}{\tau \tau}{\tau \tau}{4}{6}{6}&=-\frac{1}{2^{3/4}},
\end{aligned}
\ee
\be
\begin{aligned}\vlc{1\tau}{\tau \tau}{\tau 1}{4}{2}{4}&=-\frac{i}{\sqrt{2}},\\
\vlc{1\tau}{\tau \tau}{\tau 1}{4}{6}{4}&=\frac{1}{\sqrt{2}},
\end{aligned}
\ee
\be
\begin{aligned}\vlc{1\tau}{\tilde{\zeta }}{\zeta}{4}{3}{1}&=-\frac{i}{2^{3/4}},\\
\vlc{1\tau}{\tilde{\zeta }}{\zeta}{4}{3}{7}&=\left(-\frac{1}{2}\right)^{3/4},\\
\vlc{1\tau}{\tilde{\zeta }}{\zeta}{4}{5}{1}&=-\left(-\frac{1}{2}\right)^{3/4},\\
\vlc{1\tau}{\tilde{\zeta }}{\zeta}{4}{5}{7}&=\frac{1}{2^{3/4}},
\end{aligned}
\ee
\be
\begin{aligned}\vlc{1\tau}{\tilde{\zeta }}{\tilde{\zeta }}{4}{3}{3}&=-\frac{1}{2^{3/4}},\\
\vlc{1\tau}{\tilde{\zeta }}{\tilde{\zeta }}{4}{3}{5}&=\left(-\frac{1}{2}\right)^{3/4},\\
\vlc{1\tau}{\tilde{\zeta }}{\tilde{\zeta }}{4}{5}{3}&=\frac{\sqrt[4]{-1}}{2^{3/4}},\\
\vlc{1\tau}{\tilde{\zeta }}{\tilde{\zeta }}{4}{5}{5}&=-\frac{1}{2^{3/4}},
\end{aligned}
\ee
\be
\begin{aligned}\vlc{1\tau}{1\tau}{1\tau}{4}{4}{4}&=\frac{1}{\sqrt[4]{2}},
\end{aligned}
\ee
\be
\begin{aligned}\vlc{1\tau}{\tau 1}{\tau \tau}{4}{4}{2}&=\frac{i}{\sqrt{2}},\\
\vlc{1\tau}{\tau 1}{\tau \tau}{4}{4}{6}&=\frac{1}{\sqrt{2}},
\end{aligned}
\ee
\be
\begin{aligned}\vlc{\tau 1}{\zeta}{\tilde{\zeta }}{4}{1}{3}&=\frac{i}{2^{3/4}},\\
\vlc{\tau 1}{\zeta}{\tilde{\zeta }}{4}{1}{5}&=\frac{\sqrt[4]{-1}}{2^{3/4}},\\
\vlc{\tau 1}{\zeta}{\tilde{\zeta }}{4}{7}{3}&=-\frac{\sqrt[4]{-1}}{2^{3/4}},\\
\vlc{\tau 1}{\zeta}{\tilde{\zeta }}{4}{7}{5}&=\frac{1}{2^{3/4}},
\end{aligned}
\ee
\be
\begin{aligned}\vlc{\tau 1}{\tau \tau}{\tau \tau}{4}{2}{2}&=\frac{1}{2^{3/4}},\\
\vlc{\tau 1}{\tau \tau}{\tau \tau}{4}{2}{6}&=-\frac{i}{2^{3/4}},\\
\vlc{\tau 1}{\tau \tau}{\tau \tau}{4}{6}{2}&=\frac{i}{2^{3/4}},\\
\vlc{\tau 1}{\tau \tau}{\tau \tau}{4}{6}{6}&=\frac{1}{2^{3/4}},
\end{aligned}
\ee
\be
\begin{aligned}\vlc{\tau 1}{\tau \tau}{1\tau}{4}{2}{4}&=\frac{i}{\sqrt{2}},\\
\vlc{\tau 1}{\tau \tau}{1\tau}{4}{6}{4}&=\frac{1}{\sqrt{2}},
\end{aligned}
\ee
\be
\begin{aligned}\vlc{\tau 1}{\tilde{\zeta }}{\zeta}{4}{3}{1}&=-\frac{i}{2^{3/4}},\\
\vlc{\tau 1}{\tilde{\zeta }}{\zeta}{4}{3}{7}&=-\left(-\frac{1}{2}\right)^{3/4},\\
\vlc{\tau 1}{\tilde{\zeta }}{\zeta}{4}{5}{1}&=\left(-\frac{1}{2}\right)^{3/4},\\
\vlc{\tau 1}{\tilde{\zeta }}{\zeta}{4}{5}{7}&=\frac{1}{2^{3/4}},
\end{aligned}
\ee
\be
\begin{aligned}\vlc{\tau 1}{\tilde{\zeta }}{\tilde{\zeta }}{4}{3}{3}&=\frac{1}{2^{3/4}},\\
\vlc{\tau 1}{\tilde{\zeta }}{\tilde{\zeta }}{4}{3}{5}&=\left(-\frac{1}{2}\right)^{3/4},\\
\vlc{\tau 1}{\tilde{\zeta }}{\tilde{\zeta }}{4}{5}{3}&=\frac{\sqrt[4]{-1}}{2^{3/4}},\\
\vlc{\tau 1}{\tilde{\zeta }}{\tilde{\zeta }}{4}{5}{5}&=\frac{1}{2^{3/4}},
\end{aligned}
\ee
\be
\begin{aligned}\vlc{\tau 1}{1\tau}{\tau \tau}{4}{4}{2}&=-\frac{i}{\sqrt{2}},\\
\vlc{\tau 1}{1\tau}{\tau \tau}{4}{4}{6}&=\frac{1}{\sqrt{2}},
\end{aligned}
\ee
\be
\begin{aligned}\vlc{\tau 1}{\tau 1}{\tau 1}{4}{4}{4}&=\frac{1}{\sqrt[4]{2}},
\end{aligned}
\ee

\end{appendix}

\bibliographystyle{apsrev}
\bibliography{StringNet}
\end{document}